%% file: main.tex
\relax
\documentclass[letterpaper]{article} % DO NOT CHANGE THIS
\usepackage[nolinenums]{custom_arxiv}  % DO NOT CHANGE THIS
\usepackage{times}  % DO NOT CHANGE THIS
\usepackage{helvet} % DO NOT CHANGE THIS
\usepackage{courier}  % DO NOT CHANGE THIS
\usepackage[hyphens]{url}  % DO NOT CHANGE THIS
\usepackage{graphicx} % DO NOT CHANGE THIS
\usepackage{natbib}  % DO NOT CHANGE THIS AND DO NOT ADD ANY OPTIONS TO IT
\usepackage{caption} % DO NOT CHANGE THIS AND DO NOT ADD ANY OPTIONS TO IT
\usepackage{pifont}% http://ctan.org/pkg/pifont
\newcommand{\cmark}{\ding{51}}%
\newcommand{\xmark}{\ding{55}}%
\usepackage{enumitem}
\usepackage{thm-restate}

\usepackage{multirow}
\usepackage{multicol}
\usepackage{booktabs}
\usepackage{float}
\usepackage{siunitx}
\usepackage{xintexpr}

\usepackage{bm}
\usepackage{mleftright}
\usepackage{amsmath}
\usepackage{amssymb}
\usepackage{amsthm}
\usepackage{xcolor}
\usepackage{mathtools}
\usepackage{todonotes}
\usepackage[capitalise]{cleveref}
\usepackage{dsfont}

\newtheorem{lemma}{Lemma}

\newtheorem{proposition}{Proposition}
\newtheorem{definition}{Definition}

\newcommand{\I}{\mathcal{I}}

\newcommand{\Y}{\mathcal{Y}}

\newcommand{\bb}[1]{\bm{#1}}
\newcommand{\one}{\mathds{1}}

\newcommand{\chance}{\textsc{c}}
\newcommand{\connected}{\rightleftharpoons}
\newcommand{\seqSet}{\mathcal{Y}}
\newcommand{\rele}{\bowtie}
\newcommand{\emptyseq}{\varnothing}
\newcommand{\dnr}{\textcolor{black!80!white}{\scriptsize$\bigstar$}}

\renewcommand{\vec}[1]{\bm{#1}}

\newcommand{\vxi}{\vec{\xi}\team}

\newcommand{\vdual}{\vec{\gamma}}
\newcommand{\dual}{\gamma}

\newcommand{\cV}{\mathcal{V}}

\definecolor{col1}{HTML}{2f6c9e}
\definecolor{col2}{HTML}{d16800}

\newcommand{\tone}{_\text{\normalfont\textcolor{col1}{\sffamily T1}}}
\newcommand{\ttwo}{_\text{\normalfont\textcolor{col1}{\sffamily T2}}}
\newcommand{\team}{_\text{\normalfont\textcolor{col1}{\sffamily T}}}
\newcommand{\opp}{_\text{\normalfont\textcolor{col2}{\sffamily O}}}
\newcommand{\opppl}{\text{\textup{\textcolor{col2}{\sffamily O}}}}
\newcommand{\tonepl}{\text{\textup{\textcolor{col1}{\sffamily T1}}}} 
\newcommand{\ttwopl}{\text{\textup{\textcolor{col1}{\sffamily T2}}}}

\newcommand*\circled[1]{\tikz[baseline=(char.base)]{
            \node[shape=circle,draw,inner sep=1pt] (char) {\scriptsize #1};}}
\newcommand{\defeq}{\mathrel{:\mkern-0.25mu=}}

\DeclareMathOperator{\co}{co}

\DeclareMathOperator{\dep}{depth}
\DeclareMathOperator*{\argmax}{arg\,max}

\newcommand{\n}[1]{
  \xintifLt{#1}{1}{
      \num{\xinttheiexpr{#1 * 1000}\relax}ms%
  }{%
    \xintifLt{#1}{60}{%
      \num[round-mode=places,round-precision=2]{#1}s%
    }{% have minutes seconds and possibly hours
      \xintAssign\xintiiDivision{\xintNum{#1}}{3600}\to\hours\minutes%
      \xintAssign\xintiiDivision{\minutes}{60}\to\minutes\seconds%
      \xintiiifGt{\hours}{0}{% hours>0%
         \num{\hours}h \num[minimum-integer-digits=2]{\minutes}m %
      }{% only minutes and seconds
         \num{\minutes}m \num[minimum-integer-digits=2]{\seconds}s%
      }%
    }%
  }%
}
\newcommand{\tl}{\textcolor{gray}{$> 6$h}}
\newcommand{\na}{\textcolor{gray}{---}}
\newcommand{\unk}{\na}

\makeatletter
\newcommand{\vast}{\bBigg@{4}}
\makeatother

\title{Faster Algorithms for Optimal Ex-Ante Coordinated Collusive Strategies in Extensive-Form Zero-Sum Games}
\author{Gabriele Farina\\
Computer Science Department\\
Carnegie Mellon University\\
\texttt{gfarina@cs.cmu.edu}
\And Andrea Celli\\
DEIB\\
Politecnico di Milano\\
\texttt{andrea.celli@polimi.it}
\And Nicola Gatti\\
DEIB\\
Politecnico di Milano\\
\texttt{nicola.gatti@polimi.it}
\And Tuomas Sandholm\\
Computer Science Department, CMU\\
Strategic Machine, Inc.\\
Strategy Robot, Inc.\\
Optimized Markets, Inc.\\
\texttt{sandholm@cs.cmu.edu}
}

\usepackage[switch]{lineno}
\makeatletter
    \AtBeginDocument{%
      \@ifpackageloaded{amsmath}{%
        \newcommand*\patchAmsMathEnvironmentForLineno[1]{%
          \expandafter\let\csname old#1\expandafter\endcsname\csname #1\endcsname
          \expandafter\let\csname oldend#1\expandafter\endcsname\csname end#1\endcsname
          \renewenvironment{#1}%
                           {\linenomath\csname old#1\endcsname}%
                           {\csname oldend#1\endcsname\endlinenomath}%
        }%
        \newcommand*\patchBothAmsMathEnvironmentsForLineno[1]{%
          \patchAmsMathEnvironmentForLineno{#1}%
          \patchAmsMathEnvironmentForLineno{#1*}%
        }%
        \patchBothAmsMathEnvironmentsForLineno{equation}%
        \patchBothAmsMathEnvironmentsForLineno{align}%
        \patchBothAmsMathEnvironmentsForLineno{flalign}%
        \patchBothAmsMathEnvironmentsForLineno{alignat}%
        \patchBothAmsMathEnvironmentsForLineno{gather}%
        \patchBothAmsMathEnvironmentsForLineno{multline}%
      }{}
    }
\makeatother
\begin{document}
\twocolumn[
\maketitle
\begin{abstract}
We focus on the problem of finding an optimal strategy for a team of two players that faces an opponent in an imperfect-information zero-sum extensive-form game. Team members are not allowed to communicate during play but can coordinate before the game.  In that setting, it is known that the best the team can do is sample a profile of potentially randomized strategies (one per player) from a joint (a.k.a. correlated) probability distribution at the beginning of the game. 
In this paper, we first provide new modeling results about computing such an optimal distribution by drawing a connection to a different literature on extensive-form correlation. Second, we provide an algorithm that computes such an optimal distribution by only using profiles where only one of the team members gets to randomize in each profile. We can also cap the number of such profiles we allow in the solution. This begets an anytime algorithm by increasing the cap. We find that often a handful of well-chosen such profiles suffices to reach optimal utility for the team. This enables team members to reach coordination through a relatively simple and understandable plan. Finally, inspired by this observation and leveraging theoretical concepts that we introduce, we develop an efficient  column-generation algorithm for finding an optimal distribution for the team. We evaluate it on a suite of common benchmark games. It is three orders of magnitude faster than the prior state of the art on games that the latter can solve and it can also solve several games that were previously unsolvable. 
\end{abstract}
]
\input{text/introduction}

\input{text/preliminaries}
\input{text/tmecor}
\input{text/tmecor_xi}

\input{text/xi_structure}
\input{table_small_support}
\input{text/small_support}

\input{text/colgen}

\input{text/experiments}

\input{text/conclusions}

\section*{Broader Impact}

Enabling the computation of strong, game-theoretic strategies for imperfect-information adversarial team games has complex effects. Such technology could be used by a team of malicious players to exploit an interaction or a specific opponent. On the other hand, the technology could also be used defensively, to play in such a way as to minimize the value that can be extracted from the agent herself. Whether the technology has a positive or negative societal impact (or none) varies depending on the nature of the imperfect-information interaction and the way the technology is implemented. We believe that publishing the algorithm increases its dissemination, thereby helping even the playing field between educated expert players and ones who might be less privileged and could thus benefit more from algorithmic strategy support.

    \section*{Acknowledgments}
    This material is based on work supported by the National Science Foundation under grants IIS-1718457, IIS-1617590, IIS-1901403, and CCF-1733556, and the ARO under awards W911NF-17-1-0082 and W911NF2010081. Gabriele Farina is supported by a Facebook fellowship.

\bibliographystyle{custom_arxiv}
\bibliography{dairefs}

\clearpage
\appendix
\input{text/appendix_proofs}

\input{text/appendix_games}

\input{text/appendix_experiments}

\end{document}

%% file: text/introduction.tex
\section{Introduction}

Much of the computational game theory literature has focused on finding strong strategies for large two-player zero-sum extensive-form games. In that setting, perfect game playing corresponds to playing strategies that belong to a Nash equilibrium, and  such strategies can be found in polynomial time in the size of the game. Recent landmark results, such as superhuman agents for heads-up limit and no-limit Texas hold'em poker~\citep{Bowling15:Heads,Brown19:Superhuman,Moravvcik17:DeepStack} show that the problem of computing strong strategies in two-player zero-sum games is well understood both in theory and in practice. The same cannot be said for almost any type of strategic multi-player interaction, where computing strong strategies is generally hard in the worst case. Also, all superhuman AI gaming milestones have been in two-player zero-sum games, with the exception of multi-player no-limit Texas hold'em recently~\cite{Brown19:Superhuman}. 

In this paper, we study \emph{adversarial team games}, that is, games in which a team of coordinating (colluding) players faces an opponent. We will focus on a two-player team coordinating against a third player. Team members can plan jointly at will before the game, but are not allowed to communicate during the game (other than through their actions in the game). These games are a popular middle ground between two-player zero-sum games and multiplayer games~\citep{Stengel97:Team,Celli18:Computational}. They can be used to model many strategic interactions of practical relevance. For example, how should two players colluding against a third at a poker table play? Or, how would the two defenders in Bridge (who are prohibited from communicating privately during the game) play optimally against the declarer?
Even though adversarial team games are conceptually zero-sum interactions between two entities---the team and the opponent---computing optimal strategies is hard in this setting. Even finding a best-response strategy for the team given a fixed strategy for the opponent is hard~\cite{Celli18:Computational}.

One might think that finding the optimal strategy for the team simply amounts to finding an optimal profile of potentially mixed (a.k.a. randomized) strategies, one strategy per team members. A solution of this type that yields maximum expected sum of utilities for the team players against a rational (that is, best-responding) opponent is known as a \emph{team-maxmin equilibrium} (TME) strategy~\citep{Basilico17:Team,Zhang20:Computing,Zhang20:Converging}. 

In this paper, we are interested in a more powerful model. Before the game starts, the team members are able to sample a profile from a joint (a.k.a. correlated) \emph{distribution}. This form of \emph{ex-ante coordination} is known to be the best a team can do and comes with two major advantages. First, it offers the team larger (or equal) expected utility than TME---sometimes with dramatic gains~\citep{Celli18:Computational}. Second, it makes the problem of computing the optimal team strategy convex---and thus more amenable to the plethora of convex optimization algorithms that have been developed over the past 80 years---whereas the problem of computing a TME strategy is not convex. In our model, an optimal distribution for the team is known as a \emph{team-maxmin equilibrium with coordination device} (TMECor) strategy~\cite{Celli18:Computational,Farina18:Ex}. Finding a TMECor strategy is \textsf{NP}-hard and inapproximable~\cite{Celli18:Computational}.

We propose a new formulation for the problem of finding a TMECor strategy. In doing so, we introduce the key notion of a \emph{semi-randomized correlation plan} and draw connections with a particular strategy polytope defined by~\citet{Stengel08:Extensive}.
Second, we propose an algorithm for computing a TMECor strategy when only a fixed number of pairs of semi-randomized correlation plans is allowed. This begets an anytime algorithm by increasing that fixed number. We find that often a handful of well-chosen semi-randomized correlation plans is enough to reach optimal utility. This enables team members to reach coordination through simple and understandable strategies. 
%This would make a crucial difference if the strategy were to be executed by a human player with bounded rationality.% \textcolor{red}{that can remember only a small number of tactics}. 
%
Finally, by leveraging the  theoretical concepts that we introduce, we develop an efficient optimal column-generation algorithm for finding a TMECor strategy. We evaluate it on a suite of common benchmark games. It is three orders of magnitude faster than the prior state of the art on games that the latter can solve. It can also solve many games that were previously unsolvable.

%% file: text/preliminaries.tex
\section{Preliminaries: Extensive-Form Games}\label{sec:prel}

Extensive-form games (EFGs) are a standard model in game theory. They model games that are played on a game tree, and can capture both sequential and simultaneous moves, as well as private information. In this paper, we focus on three-player zero-sum games where two players---$\textcolor{col1}{\textnormal{\sffamily T1}}$ and $\textcolor{col1}{\textnormal{\sffamily T2}}$---play as a team against the opponent player, denoted by \opppl{}.

Each node $v$ in the game tree belongs to exactly one player $i\in\{\tonepl,\ttwopl,\opppl\}\cup\{\chance\}$ whose turn is to move.
Player $\chance$ is a special player, called the {\em chance player}. It models exogenous stochasticity in the environment, such as drawing a card from a deck or tossing a coin.
The edges leaving $v$ represent the actions available at that node. Any node without outgoing edges is called a {\em leaf} and represents an end state of the game. We denote the set of such nodes by $Z$. Each $z\in Z$ is associated with a tuple of payoffs specifying the payoff $u_i(z)$ of each player $i\in\{\tonepl,\ttwopl,\opppl\}$ at $z$.
The product of the probabilities of all actions of \chance{} on the path from the root of the game to leaf $z$ is denoted by $p_\chance(z)$.

Private information is represented via {\em information set} (infoset). In particular, the set of nodes belonging to $i\in\{\tonepl,\ttwopl,\opppl\}$ is partitioned into a collection $\I_i$ of non-empty sets: each $I\in\I_i$ groups together nodes that Player $i$ cannot distinguish among, given what they have observed. 
Necessarily, for any $I\in\I_i$ and $v,w\in I$, nodes $v$ and $w$ must have the same set of available actions. Consequently, we denote the set of actions available at all nodes of $I$ by $A_I$.
As it is customary in the related literature, we assume {\em perfect recall}, that is, no player forgets what he/she knew earlier in the game.
Finally, given players $i$ and $j$, two infosets $I_i\in \I_i$, $I_j\in\I_j$ are {\em connected}, denoted by $I_i\connected I_j$, if there exist $v\in I_i$ and $w\in I_j$ such that the path from the root to $v$ passes through $w$ or vice versa.

\noindent\textbf{Sequences. }
The set of {\em sequences} of Player $i$, denoted by $\Sigma_i$, is defined as $\Sigma_i\defeq\mleft\{(I,a):I\in\I_i,a\in A_I\mright\}\cup\{\varnothing\}$, where the special element $\varnothing$ is called the {\em empty sequence} of Player $i$.
The {\em parent sequence} of a node $v$ of Player $i$, denoted $\sigma(v)$, is the last sequence (information set-action pair) for Player $i$ encountered on the path from the root of the game to that node. Since the game has perfect recall, for each $I\in\I_i$, nodes belonging to $I$ share the same {\em parent sequence}. So, given $I\in\I_i$, we denote by $\sigma(I)\in\Sigma_i$ the unique parent sequence of nodes in $I$. Additionally, we let $\sigma(I)=\varnothing$ if Player $i$ never acts before infoset $I$.

\noindent\textbf{Relevant sequences.}
% Two information sets $(I, J) \in \I_i\times I_j$ are \emph{connected}, denoted $I \connected J$, if there exist nodes $u \in I, v \in J$ such that $u$ is on the path from the root to $v$, or vice versa.
A pair of sequences $\sigma_i\in\Sigma_i$, $\sigma_j\in\Sigma_j$ is {\em relevant} if either one is the empty sequence, or if the can be written as $\sigma_i=(I_i,a_i)$ and $\sigma_j=(I_j,a_j)$ with $I_i\connected I_j$. We write $\sigma_i \rele \sigma_2$ to denote that they form a pair of relevant sequences. Given two players $i$ and $j$, we let $\Sigma_i\rele\Sigma_j \defeq \{(\sigma_i,\sigma_j) : \sigma_i \in \Sigma_i, \sigma_j \in \Sigma_j, \sigma_i \rele\sigma_j\}$.
Similarly, given $\sigma_i$ and $I_j\in\I_j$, we say that $(\sigma_i,I_j)$ forms a relevant sequence-information set pair ($\sigma_i\rele I_j$), if $\sigma_i=\varnothing$ or if $\sigma_i=(I_i,a_i)$ and $I_i\connected I_j$.

\noindent\textbf{Reduced-normal-form plans.}
A \emph{reduced-normal-form} plan $\pi_i$ for Player $i$ defines a choice of action for every information set $I \in \I_i$ that is still reachable as a result of the other choices in $\pi$ itself. The set of reduced-normal-form plans of Player $i$ is denoted $\Pi_i$. We denote by $\Pi_i(I)$ the subset of reduced-normal-form plans that prescribe all actions for Player $i$ on the path from the root to information set $I \in \I_i$.
Similarly, given $\sigma=(I,a)\in\Sigma_i$, let $\Pi_i(\sigma)\subseteq \Pi_i(I)$ be the set of reduced-normal-form plans belonging to $\Pi_i(I)$ where Player $i$ plays action $a$ at $I$, and let $\Pi_i(\emptyseq)\defeq \Pi_i$.
Finally, given a leaf $z \in Z$, we denote with $\Pi_i(z) \subseteq \Pi_i$ the set of reduced-normal-form where Player $i$ plays so as to reach $z$.

\noindent\textbf{Sequence-form strategies.}
A \emph{sequence-form strategy} is a compact strategy representation for perfect-recall players in EFGs~\citep{Romanovskii62:Reduction,Koller96:Efficient}. Given a player $i\in\{\tonepl,\ttwopl,\opppl\}$ and a normal-form strategy $\mu\in\Delta(\Pi_i)$,\footnote{$\Delta(X)$ denotes the probability simplex over the finite set $X$.}
the sequence-form strategy induced by $\mu$ is the real vector $\bb{y}$, indexed over $\sigma\in\Sigma_i$, defined as $y[\sigma]\defeq \sum_{\pi\in\Pi_i(\sigma)}\mu(\pi)$.
The set of sequence-form strategies that can be induced as $\mu$ varies over $\Delta(\Pi_i)$ is denoted by $\seqSet_i$ and is known to be a convex polytope (called the {\em sequence-form polytope}) defined by a number of constraints equal to $|\I_i|$~\citep{Koller96:Efficient}.

%% file: text/tmecor.tex
\section{TMECor Formulation and Prior Work}
A TMECor strategy is a probability distribution $\mu\team$ over the set of randomized strategy profiles $\Y\tone\times \Y\ttwo$ that guarantees maximum expected utility for the team against the best-responding opponent \opppl{}. Since each player has perfect recall, any randomized strategy for a player is equivalent to a distribution over reduced-normal-form pure strategies~\citep{Kuhn53:Extensive}. Hence, any distribution over profiles of randomized strategies of the team members can be expressed in an equivalent way as a distribution over \emph{deterministic} strategy profiles $\Pi\tone\times\Pi\ttwo$. The benefit of this transformation is that $\Pi\tone \times\Pi\ttwo$ is a finite set, unlike $\Y\tone\times\Y\ttwo$. For this reason, TMECor is usually defined in the literature as a distribution over $\Pi\tone \times\Pi\ttwo$ without loss of generality. We will follow the same approach in our characterization.

\noindent {\bf TMECor as a Bilinear Saddle-Point Problem. }
For each leaf $z$, let $\hat{u}\team(z)\defeq (u\tone(z)+u\ttwo(z))p_\chance(z)$. The expected utility of the team can be written as the following function of the distributions of play $\mu\team \in \Delta(\Pi\tone\times\Pi\ttwo),\mu\opp \in \Delta(\Pi\opp)$:
\[
    u\team(\mu\team, \mu\opp)\!\defeq\!\sum_{z \in Z}\! \hat{u}\team(z) \!\mleft(\!\sum_{\substack{\pi\tone \in \Pi\tone(z)\\\pi\ttwo \in \Pi\ttwo(z)}\hspace{-9mm}} \hspace{-0cm}\mu\team(\pi\tone,\pi\ttwo)\!\mright)\hspace{-.1cm}\mleft(\!\sum_{{\pi \in \Pi\opp(z)}\hspace{-5.5mm}} \mu\opp(\pi)\!\mright)\!.
\]
By definition, a \emph{team-maxmin equilibrium with coordination device} (TMECor) is a Nash equilibrium of the game where the team plays according to the coordinated strategy $\mu\team\in\Delta(\Pi\tone\times\Pi\ttwo)$. In the zero-sum setting, this amounts to finding a solution of the optimization problem
\begin{equation}\label{eq:tmecor_basic}
    \displaystyle\argmax_{\mu\team \in \Delta(\Pi\tone \times \Pi\ttwo)}\min_{\mu\opp\in \Delta(\Pi\opp)} u\team(\mu\team,\mu\opp).
\end{equation}

The opponent's strategy $\mu\opp$ can be compactly represented through its equivalent sequence-form representation.
This is not the case for $\mu\team$, which cannot be represented concisely through the sequence form as shown by~\citet{Farina18:Ex}. 

%Specifically, a single meta-player comprising of all team members may have imperfect-recall, which prevents from applying the sequence form. This happens, for example, in card playing games such as Bridge, where team members have private information regarding their current hands.

\noindent {\bf Prior algorithms. } 
Prior work on the computation of TMECor mainly differs in the way the team's distribution $\mu\team$ is represented. 
\citet{Celli18:Computational} directly represent the strategy as a probability distribution over the set of joint reduced-normal-form plans $\Pi\tone\times\Pi\ttwo$. The number of bits required to store such a distribution is exponential in the size of the game tree in the worst case. They propose a column-generation approach in which, at each iteration, a new pair of pure strategies is added to the support of the distribution $\mu\team$. % Each new pair is generated through a MIP best-response oracle with a number of binary variables equal to the number of leaf nodes.
\citet{Farina18:Ex} show that it suffices to employ $|Z|$-dimensional \emph{vectors of realizations} where each $z$ is mapped to its probability of being reached when the team follows $\mu\team$. % through the mapping
    % \begin{equation}\label{eq:rho z}
    %     \rho\team[z] \defeq \sum_{\substack{\pi\tone \in \Pi\tone(z)\\ \pi\ttwo \in \Pi\ttwo(z)}} \mu\team(\pi\tone,\pi\ttwo).
    % \end{equation}
    A \emph{realization-form strategy} is a more concise representation than the original distribution $\mu\team$. The authors propose a structural decomposition of the problem which is then used to prove convergence of a fictitious-play-like algorithm. %At each iteration of the algorithm, a MIP best-response oracle generates a new vertex of the realization polytope. The oracle has a number of binary variables equal to the number of sequences of one of the two team players.

%% file: text/tmecor_xi.tex
\section{A Formulation of TMECor Based on Extensive-Form Correlation Plans}

% Si dice che l'idea che introduciamo è di usare il politopo di correlazione invece della realization form.
%
% Svantaggi:
% - Più entries nel vettore
%
% Vantaggi:
% - oracolo di best response è molto più tight
% - vale una structural decomposition simile (vedi sotto)
%
% === Remark sul fatto che solo le relevant sequences appaiono, non tutte le coppie.
%
% Overall, i vantaggi superano decisamente gli svantaggi. Inoltre, nel resto del paper mostriamo come la nuova rappresentazione e l'oracolo possono essere usati in un algoritmo diverso da quello di NeurIPS 2019:
% - Small support
% - Column generation

%As discussed, the advantage of the {realization form} introduced by \citet{Farina18:Ex} resides in the fact that it compactly represents the correlated distribution of play $\mu\team$ of the team as a vector with as many components as leaves in the game. This is in stark contrast with the naive representation of $\mu\team$ as an vector in the probability simplex with support $\Pi\tone\times\Pi\ttwo$---a set with an exponential number of elements in general. However, despite that major benefit, it is not known how to effectively characterize the realization polytope. Hence, MIP oracles exploiting the realization form are of limited scalability (see~\cref{sec:exp}).

We propose using a different representation of the correlated distribution of play $\mu\team$, inspired by the growing body of literature on extensive-form correlated equilibria. Like the \textit{realization form} by \citet{Farina18:Ex}, in our approach we represent $\mu\team$ as a vector with only a polynomial number of components. However, unlike the realization form, the number of components scales  as the product of the number of sequences of the two players, which can be significantly larger than the number of leaves. This downside is amply outweighed by the following benefits.
First, we show that in practice our proposed representation of $\mu\team$ enables us to compute best responses for the team significantly faster than the prior representations.
Second, in certain classes of games, we even show that our proposed representation enables the computation of a TMECor in polynomial time. This is the case, for example, in Goofspiel, a popular benchmark game in computational game theory~\citep{Ross71:Goofspiel}.

\subsection{Extensive-Form Correlation Plans}

Our representation is based on the concept of \emph{extensive-form correlation plans}, introduced by \citet{Stengel08:Extensive} in their seminal paper on extensive-form correlation. In particular, we map the correlated distribution of play $\mu\team$ of the team to the vector $\vec{\xi}\team$ indexed over pairs of sequences $(\sigma\tone,\sigma\ttwo) \in \Sigma\tone\rele\Sigma\ttwo$, where each entry is defined as
\begin{equation}\label{eq:xi def}
	\xi\team[(\sigma\tone,\sigma\ttwo)]\defeq \sum_{\substack{\pi\tone\in\Pi\tone(\sigma\tone)\\\pi\ttwo\in\Pi\ttwo(\sigma\ttwo)}} \mu\team[(\pi\tone,\pi\ttwo)]. 
\end{equation}

Here $\vec{\xi}\team$ is not indexed over \emph{all} pairs  of sequences $(\sigma\tone,\sigma\ttwo)$---only \emph{relevant} sequence pairs. While there are games in which this distinction is meaningless (that is, games in which all sequences pairs for the team members are relevant), in practice the number of relevant sequence pairs is only a tiny fraction of the total number of sequence pairs, as shown in \cref{table:rele vs prod}(b).

\begin{table*}%[th]
    \centering
    \sisetup{scientific-notation=false,round-mode=places,
round-precision=2}
    \setlength{\tabcolsep}{4.0pt}\small\centering
    \begin{tikzpicture}
    \node[anchor=west] at (0,0) {
    \begin{tabular}{lrrrr|||rr|||ccc}
        \toprule
        \multirow{2}{*}{\makebox[.6cm]{~}\bf Game instance} & \multicolumn{3}{c}{\bf Num. sequences} & \bf Num. leaves & \multirow{2}{*}{\scriptsize{$\displaystyle\frac{|\Sigma\tone \rele \Sigma\ttwo|}{|Z|}$}} & \multirow{2}{*}{\scriptsize{$\displaystyle\frac{|\Sigma\tone \times \Sigma\ttwo|}{|\Sigma\tone \rele \Sigma\ttwo|}$}} & \multicolumn{3}{c}{\bf Triangle-free?} \\
        &$|\Sigma_1|$ & $|\Sigma_2|$ & $|\Sigma_3|$ & \multicolumn{1}{c}{$|Z|$}&&&\small$\opppl{}=1$&\small$\opppl{}=2$&\small$\opppl{}=3$\\
        \midrule
        \makebox[.6cm][l]{\textbf{[A]}}Kuhn poker (3 ranks) &\num{25}&\num{25}&\num{25} &  \num{78} & \num{3.3974} & \num{2.358} &\xmark&\xmark&\xmark\\
        \makebox[.6cm][l]{\textbf{[B]}}Kuhn poker (4 ranks) &\num{33}&\num{33}&\num{33} & \num{312} & \num{1.5929} & \num{2.1911} &\xmark&\xmark&\xmark \\
        \makebox[.6cm][l]{\textbf{[C]}}Kuhn poker (12 ranks) &\num{97}&\num{97}&\num{97} & \num{17160} & \num{0.288} & \num{1.9027} &\xmark&\xmark&\xmark\\
        \midrule
        \makebox[.6cm][l]{\textbf{[D]}}Goofspiel (3 ranks, limited info) & \num{934} & \num{934} & \num{934} & \num{1296} & \num{9.5355} & \num{70.588} &\cmark&\cmark&\cmark\\
        \makebox[.6cm][l]{\textbf{[E]}}Goofspiel (3 ranks) &\num{1630}&\num{1630}&\num{1630} & \num{1296} & \num{15.53} & \num{131.96} &\cmark&\cmark&\cmark \\
        \midrule
        \makebox[.6cm][l]{\textbf{[F]}}Liar's dice (3 faces) &\num{1021}&\num{1021}&\num{1021} & \num{13797} & \num{5.26} & \num{14.42} &\xmark&\xmark&\xmark \\
        \makebox[.6cm][l]{\textbf{[G]}}Liar's dice (4 faces) &\num{10921}&\num{10921}&\num{10921} & \num{262080} & \num{6.2518} & \num{72.79} &\xmark&\xmark&\xmark\\
        \midrule
        \makebox[.6cm][l]{\textbf{[H]}}Leduc poker (3 ranks, 1 raise) &\num{457}&\num{457}&\num{457} & \num{4500} & \num{2.6215} & \num{17.703568703907773} &\xmark&\xmark&\xmark\\
        \makebox[.6cm][l]{\textbf{[I]}}Leduc poker (4 ranks, 1 raise) & 801 &801&801 & 16908 & \num{1.338123964986988} & \num{28.358055248618785} &\xmark&\xmark&\xmark\\
        \makebox[.6cm][l]{\textbf{[J]}}Leduc poker (2 ranks, 2 raises) & 1443 & 1443 & 1443 & 3786 & \num{7.275488642366614} & \num{75.594445452895262} &\xmark&\xmark&\xmark\\
        %Leduc poker (3 raises) &\num{7687}&\num{7687}&\num{7687} & \num{81072} && \num{\xintfloateval{357802 / 81072}} & \num{\xintfloateval{59089969 / 357802}} &&\xmark&\xmark&\xmark \\
        \bottomrule
    \end{tabular}};
    \draw[line width=2.3mm,white] (10.54,-2.6) -- (10.54,2.6);
    \draw[line width=2.3mm,white] (14.28,-2.6) -- (14.28,2.6);
    \node[anchor=north] at (5.1,-2.55) {\textbf{(a)} --- Game instances and sizes};
    \node[anchor=north] at (12.5,-2.55) {\textbf{(b)}};
    \node[anchor=north] at (16.0,-2.55) {\textbf{(c)}};
    \end{tikzpicture}\vspace{-2mm}
    \caption{\textbf{(a)} Size of the game instances used in our experiments, in terms of number of sequences $|\Sigma_i|$ for each player $i$, and number of leaves $|Z|$.
    \textbf{(b)} Ratio between the number of leaves $|Z|$, number of sequence pairs for the team members $|\Sigma\tone \times \Sigma\ttwo|$, and number of \emph{relevant} sequence pairs for the team members $|\Sigma\tone \rele \Sigma\ttwo|$ in various benchmark games. For all games reported in the subtable, we chose the first two players to act as the team members. \textbf{(c)} The subtable reports whether the interaction of the team members is triangle-free \citep{Farina20:Polynomial}, given the opponent player \opppl{}.}
    \label{table:rele vs prod}\vspace{-3mm}
\end{table*}

The set of extensive-form correlation plans $\vec{\xi}\team$ that can be induced as $\mu\team$ varies over the set of all correlated distributions of play for the team members is a convex polytope. We denote it as $\Xi\team$ and call it the \emph{polytope of correlation plans}. We will recall existing results and provide new ones about the structure of $\Xi\team$ in \cref{sec:srcp}. 

\subsection{Computing a TMECor using Correlation Plans}

Extensive-form correlation plans encode a \emph{superset} of the information encoded by realization plans. Indeed,  for all $z$,
$\xi\team[\sigma\tone(z),\sigma\ttwo(z)] = \rho\team[z].$
Using the previous identity, we can rewrite the problem of computing a TMECor of a constant-sum game \eqref{eq:tmecor_basic} as
\[
    \argmax_{\vec{\xi}\team \in \Xi\team} \min_{\vec{y}\opp \in \Y\opp} \sum_{z\in Z} \hat{u}\team(z) \xi\team[\sigma\tone(z),\sigma\ttwo(z)] y[\sigma\opp(z)].
\]
By dualizing the inner linear minimization problem over $\vec{y}\opp$, we get the following proposition that shows that a TMECor can be found as the solution to a linear program (LP) with a polynomial number of variables. (All the proofs of this paper can be found in the appendix.)
\begin{restatable}{proposition}{proptmecorlp}\label{prop:tmecor as lp}
    An extensive-form correlation plan $\vec{\xi}\team$ is a TMECor if and only if it is a solution to the LP
    \begin{equation*}
        \mleft\{\hspace{-1.25mm}\begin{array}{l}
            \displaystyle
            %% OBJECTIVE
            \argmax_{\vec{\xi}\team}~~~ v_\emptyseq,\quad\text{\normalfont subject to:}\\
            %\text{\normalfont subject to:}\\
            %% CONSTRAINT 1
            ~\circled{\normalfont 1}~ \displaystyle v_I - ~\sum_{\mathclap{\substack{I' \in \mathcal{I}\opp\\\sigma\opp(I') = (I, a)}}}~ v_{I'} \le ~\sum_{\mathclap{\substack{z \in Z\\\sigma\opp(z) = (I,a)}}} \hat{u}\team(z)\xi\team[\sigma\tone(z),\sigma\ttwo(z)]\\[-4mm]
            \hspace{5cm}\forall\, (I,\!a) \!\in\! \Sigma\opp\!\!\setminus\! \{\emptyseq\}\\[2mm]
            %% CONSTRAINT 2
            ~\circled{\normalfont 2}~\displaystyle v_\emptyseq - ~\sum_{\mathclap{\substack{I' \in \mathcal{I}\opp\\\sigma\opp(I') = \emptyseq}}}~ v_{I'} \le ~\sum_{\mathclap{\substack{z \in Z\\\sigma\opp(z) = \emptyseq}}} \hat{u}\team(z)\xi\team[\sigma\tone(z),\sigma\ttwo(z)]\\[8mm]
            %% CONSTRAINT 3
            ~\circled{\normalfont 3}~v_\emptyseq\textnormal{ free}, v_{I} \textnormal{ free } \quad \forall\ I \in \mathcal{I}\opp\\[.5mm]
            %% CONSTRAINT 4
            ~\circled{\normalfont 4}~\vec{\xi}\team \in \Xi\team.
        \end{array}\mright.
    \end{equation*}
\end{restatable}

\noindent As a direct consequence of \cref{prop:tmecor as lp}, a TMECor can be found in polynomial time whenever $\Xi\team$ can be represented as the intersection of a set of polynomially many linear constraints. In \cref{sec:srcp}, we recall when that is the case.

%% file: text/xi_structure.tex
\vspace{-1mm}
\section{Semi-Randomized Correlation Plans and the Structure of $\Xi\team$}\label{sec:srcp}

Even though $\Xi\team$ is a convex polytope, the set of (potentially exponentially many) linear constraints that define it is not known in general. So, alternative characterizations of the set $\Xi\team$ are needed before the LP in \cref{prop:tmecor as lp} can be solved. In this section, we recall two known results about the structure of $\Xi\team$, and propose a new one (\cref{prop:convex hull}). We will use our result to arrive at two different approaches to tackle the LP of \cref{prop:tmecor as lp} in \cref{sec:small_support} and~\ref{sec:colgen}, respectively.

\vspace{-1mm}
\subsection{Containment in the von Stengel-Forges Polytope}
The first result about the structure of $\Xi\team$ has to do with a particular polytope that was introduced by \citet{Stengel08:Extensive}.
\begin{definition}\label{def:vsf}
    The \emph{von Stengel-Forges polytope} of the team, denoted $\cV\team$, is the polytope of all vectors $\vec{\xi} \in \mathbb{R}_{\ge 0}^{|\Sigma\tone\rele\Sigma\ttwo|}$ indexed over relevant sequence pairs that satisfy the following polynomially-sized set of linear constraints.\vspace{-1mm}
    \[
        \begin{array}{l}
        \circled{\normalfont 1}\quad\displaystyle \xi[\emptyseq,\emptyseq] = 1\\[1mm]
        \circled{\normalfont 2}\quad\displaystyle\sum_{\mathclap{a\tone \in A_{I\tone}}} \xi[(I\tone,a\tone),\sigma\ttwo] = \xi[\sigma(I\tone),\sigma\ttwo] \quad\forall I\tone \rele\sigma\ttwo\\[4mm]
        \circled{\normalfont 3}\quad\displaystyle\sum_{\mathclap{a\ttwo \in A_{I\ttwo}}} \xi[\sigma\tone, (I\ttwo,a\ttwo)] = \xi[\sigma\tone, \sigma(I\ttwo)] \quad\forall \sigma\tone \rele I\ttwo.
        \end{array}
    \]
\end{definition}
\vspace{-1mm}
\noindent These can be interpreted as ``probability mass conservation'' constraints. They are interlaced sequence-form constraints. 

The following result by \citet{Stengel08:Extensive} is immediate from the definition of $\xi\team$ in \eqref{eq:xi def}.
\begin{proposition}[\citet{Stengel08:Extensive}]\label{prop:xi subset vsf}
    The set of extensive-form correlation plans is a subset of the von Stengel-Forges polytope. Formally, $\Xi\team \subseteq \cV\team$.
\end{proposition}

\subsection{Triangle-Freeness and Polynomial-Time Computation of TMECor}\vspace{-1mm}
\cref{prop:xi subset vsf}
%, due to \citet{Stengel08:Extensive}, 
shows that $\Xi\team$ is a subset of the von Stengel-Forges polytope. There are games where the reverse inclusion does not hold. \citet{Farina20:Polynomial} gave a sufficient condition---called \emph{triangle-freeness}---for the reverse inclusion to hold. We state the condition for our setting.
\begin{definition}[\citet{Farina20:Polynomial}]
    The interaction of the team members \tonepl{} and \ttwopl{} is \emph{triangle-free} if, for any choice of distinct information sets $I_1, I_2 \in \I\tone$ with $\sigma\tone(I_1) = \sigma\tone(I_2)$ and any choice of distinct information sets $J_1, J_2 \in \I\ttwo$ with $\sigma\ttwo(J_1) = \sigma\ttwo(J_2)$, it is never the case that $(I_1 \connected J_1) \land (I_2 \connected J_2) \land (I_1 \connected J_2)$.
\end{definition}

\noindent\citet{Farina20:Polynomial} show that when the information structure of correlating players (in our case, the team members) is triangle-free, then $\Xi\team = \cV\team$. So, when the interaction of the team is triangle-free, a TMECor can be found in polynomial time by substituting constraint \circled{4} in the LP in \cref{prop:tmecor as lp} with the von Stengel-Forges constraints of \cref{def:vsf}. As far as we are aware, this positive complexity result has not been noted before in the literature. We show in \cref{table:rele vs prod}(c) that Goofspiel is triangle free (and that none of the other common benchmark games that we consider are).

\subsection{Semi-Randomized Correlation Plans} %  and Convex-Hull Characterization

We now give a third result about the structure of $\Xi\team$, which will enable us to replace Constraint \circled{4} of \cref{prop:tmecor as lp} with something more practical. First, we introduce \emph{semi-randomized correlation plans}, which are subsets of the von Stengel-Forges polytope of the team. They represent strategy profiles in which one of the players plays a deterministic strategy, while the other player in the team independently plays a randomized strategy. Formally, we define the set of semi-randomized correlation plans for \tonepl{} and \ttwopl{} as
    \begin{align*}
        \Xi\tone^* &= \{\vec{\xi} \in \cV\team : \xi[\emptyseq, \sigma\ttwo] \in \{0,1\} \quad \forall\ \sigma\ttwo \in \Sigma\ttwo\}, \\
        \Xi\ttwo^* &= \{\vec{\xi} \in \cV\team : \xi[\sigma\tone, \emptyseq] \in \{0,1\} \quad \forall\ \sigma\tone \in \Sigma\tone\},
    \end{align*}
respectively. Crucially, a point $\vec{\xi} \in \Xi_i^*$ for $i\in \{\tonepl,\ttwopl\}$ can be expressed using real and binary variables, in addition to the linear constraints the define $\cV$ (\cref{def:vsf}).

With that, we can show the following structural result for the polytope of extensive-form correlation plans $\Xi\team$.

\begin{restatable}{proposition}{propconvexhull}\label{prop:convex hull}
    In every game, $\Xi\team$ is the convex hull of the set $\Xi\tone^*$, or equivalently of the set $\Xi\ttwo^*$. Formally, $\Xi\team = \co \Xi\tone^*  = \co \Xi\ttwo^* = \co (\Xi\tone^* \cup \Xi\ttwo^*)$.
\end{restatable}
%IN THE JOURNAL VERSION, ALSO INCLUDE THE OTHER VERSION OF THIS PROPOSITION THAT HAS THE INTEGRALITY CONSTRAINTS AND THAT LEADS TO ALGORITHMS ONLY HAVE PROFILES OF PURE STRATEGIES IN THE SUPPORTS. EXPLAIN WHY THAT IS WORSE: THE NUMBER OF VARIABLES IN THE PRICING MIP IS QUADRATIC INSTEAD OF LINEAR. ??? WOULD THE COLUMN GENERATION TAKE MORE ITERATIONS BECAUSE THE VARIABLES THAT ARE GETTING ADDED WOULD BE PURE AND THUS LESS FLEXIBLE SO YOU MIGHT NEED MORE OF THEM BEFORE REACHING OPTIMALITY? ???

%As we show in the next two sections, we use \cref{prop:convex hull} in two different ways.

% Advantage: $\vxi \in \Xi\tone^*$ can be done with a few linear constraints and integer variables.

%% file: table_small_support.tex
\begin{table*}
    \sisetup{detect-all=true,scientific-notation=fixed,fixed-exponent=0,round-mode=places,
round-precision=4}
\setlength{\tabcolsep}{4.5pt}\small\centering
\begin{tikzpicture}
\node[anchor=south west] at (0, 0) {
\begin{tabular}{cc||rrr|r||rrr|r||rrr|r}
    \toprule
    \multicolumn{2}{c||}{\multirow{2}{*}{\bf Game}} & \multicolumn{4}{c||}{\bf Opponent player \opppl{} = 1} & \multicolumn{4}{c||}{\bf Opponent player \opppl{} = 2} & \multicolumn{4}{c}{\bf Opponent player \opppl{} = 3}\\
    \multicolumn{2}{c||}{}&$n=1$&$n=2$&$n=3$&$n=\infty$&$n=1$&$n=2$&$n=3$&$n=\infty$&$n=1$&$n=2$&$n=3$&$n=\infty$\\
    \midrule
    \multirow{3}{*}{\begin{minipage}{.6cm}\centering Kuhn\\poker\end{minipage}} & \bf [A]
           & \num{0} & \dnr & \dnr & \num{0}    % OPP = 1 
           & \num{0} & \dnr & \dnr & \num{0}    % OPP = 2
           & \num{0} & \dnr & \dnr & \num{0}    % OPP = 3
    \\
    &\bf [B]
           & \num{2.083333333333e-02} & \num{3.787878787879e-02} & \dnr & \num{3.787878787879e-02}   % OPP = 1
           & \num{1.811594202899e-03} & \num{2.457264957265e-02} & \num{2.651515151515e-02} & \num{0.026515} % OPP = 2
           & \num{-4.166666666667e-02} & \dnr & \dnr & \num{-0.041667}   % OPP = 3
    \\
    &\bf [C]
           & \num{4.696969696970e-02} & \num{6.554988612408e-02} & \num{6.632787698421e-02} & \num{0.066403}  % OPP = 1
           & \num{1.275966239995e-02} & \num{3.671224860636e-02} & \num{3.760965149463e-02} & \num{ 0.037952} % OPP = 2
           & \num{-2.272727272727e-02} & \num{-1.532305776584e-02} & \num{-1.406490288779e-02} & \num{ -0.014013} % OPP = 3
    \\
    \midrule
    \multirow{2}{*}{\!\!Goofspiel\!\!}&\bf [D]
           & \num{0.238888893519e+00} & \num{0.252421652422e+00} & \dnr & \num{0.252422}  % OPP = 1
           & \num{0.23888888} & \num{0.25242165} & \dnr & \num{0.252422}  % OPP = 2
           & \num{0.23888888} & \num{0.252421652} & \dnr & \num{0.252422}
    \\
    &\bf [E]
           & \num{0.23888888} & \num{0.25343915} & \dnr & \num{0.253439}  % OPP = 1
           & \num{0.23888888} & \num{0.253439153439} & \dnr & \num{0.25343}     % OPP = 2
           & \num{0.23888888} & \num{0.253439153439} & \dnr & \num{0.25343}     % OPP = 3
    \\
    \midrule
    \multirow{2}{*}{\begin{minipage}{.6cm}\centering Liar's\\dice\end{minipage}}&\bf [F]
           & \num{0} & \dnr & \dnr & \num{0}   % OPP = 1
           & \num{2.098765432099e-01} & \num{2.553606237817e-01} & \num{2.561728395062e-01} & \num{0.256173}    % OPP = 2
           & \num{2.716049382716e-01} & \num{2.839506172840e-01} & \dnr & \num{0.283951}
    \\
    &\bf [G]
           & \num{6.25e-02} & \dnr & \dnr & \num{0.0625}  % OPP = 1
           & \num{2.500000000000e-01} & \num{0.2656250} & \num{0.2656250} & \unk % OPP = 2
           & \num{0.26563} & \unk & \unk & \unk % OPP = 3
    \\
    \midrule
    \multirow{3}{*}{\begin{minipage}{.6cm}\centering Leduc\\poker\end{minipage}}&\bf [H]
           & \num{1.453111360519e-01} & \num{2.245958829529e-01} & \num{0.2465675} %<<<=== double check
             & \num{0.276540}   % OPP = 1
           & \num{2.107491769076e-01} & \num{0.28627} & \num{3.142783655585e-01} & \num{0.345024} % OPP = 2
           & \num{1.840402599314e-01} & \num{2.448037777400e-01} & \num{0.28150} & \num{0.292622} % OPP = 3
    \\
    &\bf [I]
           & \unk & \unk & \unk & \num{0.1421806}  % OPP = 1
           & \unk & \unk & \unk & \num{0.1419778}  % OPP = 2
           & \unk & \unk & \unk & \num{0.0850132}  % OPP = 3
    \\
    &\bf [J]
           & \num{2.448979591837e-01} & \num{7.037037037037e-01} & \num{7.974590989826e-01} & \num{0.835897}  % OPP = 1
           & \num{2.101359703337e-01} & \num{9.222222222222e-01} & \num{0.96946} & \num{0.970940}  % OPP = 2
           & \num{2.448979591837e-01} & \num{7.037037037038e-01} & \num{0.7974591} & \num{0.8358974}  % OPP = 3
    \\
    \bottomrule
\end{tabular}};
\draw[line width=2.2mm,white] (2.26, 0.08) -- (2.26, 5.30);
\draw[line width=2.2mm,white] (7.25, 0.08) -- (7.25, 5.30);
\draw[line width=2.2mm,white] (12.18, 0.08) -- (12.18, 5.30);
\end{tikzpicture}\vspace{-3mm}
\caption{Expected utility of the team for varying support sizes ($n$). All values for $n \in \{1,2,3\}$ were computed using the MIP of \cref{sec:small_support}, while the values corresponding to $n=\infty$ were computed using our column generation approach (\cref{sec:colgen}). `$\bigstar$`: A provably optimal utility has already been obtained with a lower value of the support size $n$. `---`: We were unable to compute the exact value, because the corresponding algorithm hit the time limit.}
\label{tab:smallsupp}\vspace{-4mm}
\end{table*}

%% file: text/small_support.tex
\section{Computing TMECor with a Small Support of Semi-Randomized 
%Correlation %PUT THIS BACK IN JOURNAL VERSION???
Plans of Fixed Size}\label{sec:small_support}

From \cref{prop:convex hull}, it is known that $\Xi\team$ is the convex hull of $\Xi\tone^*$ and $\Xi\ttwo^*$.  Furthermore, the polytopes $\Xi\tone^\ast$ and $\Xi\ttwo^\ast$ can be described via a number of linear constraints that is quadratic in the game size and a number of integer variables that is linear in the game size. 
So, we can replace Constraint \circled{4} in \cref{prop:tmecor as lp} with the constraint that $\vxi$ be a convex combination of elements from $\Xi\tone^*$ and $\Xi\ttwo^*$. We introduce variables $\vxi^{(1)},\ldots,\vxi^{(n)} \in \Xi\tone^* \cup \Xi\ttwo^*$ and the corresponding convex combination coefficients $\lambda^{(1)},\ldots,\lambda^{(n)}$, and replace Constraint \circled{4} with the linear constraint $\vec{\xi}\team = \sum_{i=1}^n \lambda^{(i)}\vec{\xi}\team^{(i)}$. Here, $n$ is a parameter with which we can cap the number of semi-randomized correlation plans that can be included in the strategy.  This gives the following mixed integer LP. 
    \begin{equation*}
        \mleft\{\hspace{-1.25mm}\begin{array}{l}
            \displaystyle
            %% OBJECTIVE
            \argmax_{\vxi^{(1)},\ldots,\vxi^{(n)},\lambda^{(1)},\ldots,\lambda^{(n)}}~~~ v_\emptyseq,\quad\text{\normalfont subject to:}\\[7mm]
            %\text{\normalfont subject to:}\\
            %% CONSTRAINTS 1-3
            ~\text{constraints }\circled{\normalfont 1}~ \circled{\normalfont 2}~ \circled{\normalfont 3}\text{ as in \cref{prop:tmecor as lp}}\\[1mm]
            %% CONSTRAINT 4
            ~\circled{\normalfont 4}~\vec{\xi}\team = \sum_{i=1}^n \lambda^{(i)}\vec{\xi}\team^{(i)}\\[1mm]
            %% CONSTRAINT 5
            ~\circled{\normalfont 5}~\vec{\xi}\team^{(1)}\in\Xi\tone^*, \vec{\xi}\team^{(2)}\in\Xi\ttwo^*, \vec{\xi}\team^{(3)}\in\Xi\tone^*, \vec{\xi}\team^{(4)}\in\Xi\ttwo^*, \dots^\ddagger\\[1mm]
            %% CONSTRAINT 6
            ~\circled{\normalfont 6}~\sum_{i=1}^n \lambda^{(i)} = 1, \ \lambda^{(i)} \ge 0\quad \forall i\in\{1,\dots,n\}.
        \end{array}\mright.
    \end{equation*}
    \def\thefootnote{$\ddagger$}
    \footnotetext{In Constraint \circled{5} we alternate the set of semi-randomized correlation plans (i.e., we alternate which player's turn it is to play a deterministic strategy). Empirically, this increases the diversity of the strategies of $\Xi\team$ that can be represented with small values of $n$ and leads to higher utilities for the team.}
    \renewcommand*{\thefootnote}{\arabic{footnote}}

The larger $n$ is, the higher the solution value obtained, but the slower the program. We can make this into an anytime algorithm by solving the integer program for increasing values of $n$. 
By Caratheodory's theorem, this program already yields an optimal solution to the LP in \cref{prop:tmecor as lp} when $n \ge |\Sigma_1\rele \Sigma_2| + 1$.
As we show in detail in \cref{sec:exp}, in practice we found that near-optimal coordination can be achieved through strategies with a significantly smaller value of $n$. %---even $n = 2$ or $n=3$.
Hence, oftentimes the team does not need a large number of complex profiles of randomized strategies to play optimally: a handful of carefully selected simple  strategies often result in optimal coordination.

%% file: text/colgen.tex
\section{A Fast Column Generation Approach}\label{sec:colgen}
%In \cref{sec:small_support} we showed how to leverage \cref{prop:convex hull} to first write a TMECor as a convex combination of a \emph{fixed} number of semi-randomized correlation plans, and later optimize over the support and coefficients of the convex combination using a mixed a mixed integer LP. 
In this section, we show a different approach to solving the LP in \cref{prop:tmecor as lp}---using column generation~\citep{Ford58:Suggested}.
First, we proceed with a \emph{seeding} phase. We pick a set $S$ containing one or more points $\vxi^{(1)}, \vxi^{(2)}, \dots, \vxi^{(m)}$ that are known to belong to $\Xi\team$. %Clearly, the convex hull $\co S$ of $S$ is a subset of $\Xi\team$.
Then, the main loop starts. First, for  $i\in\{1,\dots, |S|\}$, let
\[
\beta^{(i)}(\sigma\opp)\defeq \sum_{\mathclap{\substack{z \in Z\\\sigma\opp(z) = \sigma\opp}}} \hat{u}\team(z)\xi\team^{(i)}[\sigma\tone(z),\sigma\ttwo(z)]\quad\forall\ \sigma\opp\in\Sigma\opp.
\] 
Then we solve the LP of \cref{prop:tmecor as lp} where Constraint \circled{4} has been substituted with $\vxi \in \co S$:
\begin{equation*}
    (\ast):\mleft\{\hspace{-1.25mm}\begin{array}{l}
        \displaystyle
        %% OBJECTIVE
        \argmax_{\lambda^{(1)},\dots,\,\lambda^{(|S|)}}~~~ v_\emptyseq,\quad\text{\normalfont subject to:}\\
        %\text{\normalfont subject to:}\\
        %% CONSTRAINT 1
        ~\circled{\normalfont 1}~ \displaystyle v_I - ~\sum_{\mathclap{\substack{I' \in \mathcal{I}\opp\\\sigma\opp(I') = \sigma\opp}}}~ v_{I'} - ~\sum_{i=1}^{|S|}\beta^{(i)}(\sigma\opp)\,\lambda^{(i)} \le 0\\[-5mm]
        \hspace{4.8cm}\forall\, \sigma\opp \!\in\! \Sigma\opp\!\!\setminus\! \{\emptyseq\}\\[1mm]
        %% CONSTRAINT 2
        ~\circled{\normalfont 2}~\displaystyle v_\emptyseq - ~\sum_{\mathclap{\substack{I' \in \mathcal{I}\opp\\\sigma\opp(I') = \emptyseq}}}~ v_{I'} - ~\sum_{i=1}^{|S|}\beta^{(i)}(\emptyseq)\,\lambda^{(i)} \le 0\\[9mm]
        %% CONSTRAINT 3
        \circled{\normalfont 3}~\sum_{i=1}^{|S|}\lambda^{(i)}=1\\[1mm]
        %% CONSTRAINT 4
        \circled{\normalfont 4}~\lambda^{(i)}\ge 0 \quad\forall\ i \in \{1,\dots,|S|\}\\[1mm]
        %% CONSTRAINT 5
        \circled{\normalfont 5}~v_\emptyseq\text{ free}, v_{I} \text{ free } \quad \forall\ I \in \mathcal{I}\opp.
    \end{array}\mright.
\end{equation*}
This is called the \textit{master LP}.\footnote{In $(*)$ the convex combination is among \emph{given} correlation plans, while in the MIP of \cref{sec:small_support}, the elements to combine are themselves variables.}

Given the solution to the master LP, a \emph{pricing problem} is created. The goal of the pricing problem is to generate a new element $\vxi^{|S|+1}$ to be added to $S$ so as to increase the team utility in the next iteration, that is, the next solve of the master LP that then has an additional variable. This main loop of solving the larger and larger master LP keeps repeating until termination (discussed later).

\subsection{The Pricing Problem}

The pricing problem consist of finding a correlation plan $\hat{\vec{\xi}}\team\in\Xi\team$ which, if included in the convex combination computed by $(*)$, would lead to the maximum gradient of the objective (that is, the maximum \emph{reduced cost}). By exploiting the theory of linear programming duality, such a correlation plan can be computed starting from the solution of the dual of $(*)$. In particular, let $\vdual$ be the $|\Sigma\opp|$-dimensional vector of dual variables corresponding to Constraints \circled{\normalfont 1} and \circled{\normalfont 2} of $(*)$, and $\dual' \in\mathbb{R}$ be the dual variable corresponding to Constraint \circled{\normalfont 3}. Then, the reduced cost of any candidate $\hat{\vec{\xi}}\team$ is
\begin{align*}
    &c(\hat{\vec{\xi}}\team)\defeq -\dual'+\sum_{\substack{z\in Z}} \hat{u}\team(z)\hat{\xi}\team[\sigma\tone(z),\sigma\ttwo(z)] \dual[\sigma\opp(z)].
\end{align*}

Now comes our crucial observation. Since $c(\hat{\vec{\xi}}\team)$ is a linear function, and since from \cref{prop:convex hull} we know that $\Xi\team = \co \Xi^*\tone$, by convexity
\[
    \max_{\hat{\vec{\xi}}\team \in \Xi\team} c(\hat{\vec{\xi}}\team) = \max_{\hat{\vec{\xi}}\team \in \Xi\tone^*} c(\hat{\vec{\xi}}\team).
\]
We want to solve the LP on the left hand side, but---as discussed in \cref{sec:srcp}---the constraints defining $\Xi\team$ are not known. The above equality enables us to solve the problem because the right hand side is a well-defined mixed integer LP (MIP). We can use a commercial solver such as Gurobi to solve it. %As we will show in detail in \cref{sec:exp}, our formulation of the pricing problem using semi-randomized correlation plans is extremely tight.%\footnote{Are we only adding pure profiles? Does this depend on whether the LP solver always returns a vertex? ???}
When the objective value of the pricing problem is non-positive, there is no variable that can be added to the master LP which would increase its value. Thus, the optimal solution to the master LP is guaranteed to be optimal for the LP in \cref{prop:tmecor as lp} and the main loop terminates.

\vspace{-1mm}
\subsection{Implementation Details}\label{sec:implementation_colgen}
We further speed up the solution of the pricing problem in our implementation by the following techniques.

\noindent\textbf{Seeding phase.} 
To avoid having to go through many iterations of the main loop, each of which requires solving the pricing problem, we want to seed the master LP up front with a set of good candidate variables. While any seeding maintains optimality of the overall algorithm, seeding it with variables that are likely to be part of the optimal solution increases speed the most. We initialize the set of correlation plans $S$ by running $m$ iterations of a self-play no-external-regret algorithm. Specifically, we let each player run CFR+~\cite{Tammelin15:Solving,Bowling15:Heads} and, at each iteration of that algorithm, we sample a pair of pure normal-form plans for the two team members according to the current strategies of the two players. At each iteration of that no-regret method, we set the utility of each team member to $u\tone + u\ttwo$. Finally, for each pair $(\pi\tone, \pi\ttwo)\in\Pi\tone\times\Pi\ttwo$ of normal-form plans generated by that no-regret algorithm, we compute and add to $S$ the correlation plan corresponding to the distribution $\mu$ that assigns probability $1$ to $(\pi\tone,\pi\ttwo)$ using \cref{eq:xi def}. While self-play no-regret methods guarantee convergence to Nash equilibrium in two-player zero-sum game, no guarantee is available in our setting. However, we empirically find that this seeding strategy leads to a strong initial set of correlation plans. 

\noindent\textbf{Linear relaxation.} 
Before solving the MIP formulation of the pricing problem, we first try to solve its linear relaxation
$ \argmax_{\hat{\vec{\xi}}\team\in\cV\team}c(\hat{\vec{\xi}}\team)$.
We found that in many cases it outputs semi-randomized correlation plans, thus avoiding the overhead of having to solve a MIP. 

\noindent\textbf{Solution pools.} 
Modern commercial MIP solvers such as Gurobi keep track of additional suboptimal feasible solutions (in addition to the optimal one) that were found during the process of solving a MIP. Since accessing those additional solutions is essentially free computationally, we add to $S$ all the solutions (even suboptimal ones) that were produced in the process of solving the MIP. This can be viewed as a form of dynamic seeding and does not affect the optimality of the overall algorithm.

\noindent\textbf{Termination.} 
Because fast integer and LP solvers work with real-valued variables, near the end of the column-generation loop the new variables that are generated in the pricing problem have reduced costs that are very close to zero. It is not clear whether they are actually positive or zero. Therefore, we set the numeric tolerance so that we stop the column-generation loop if the value of the pricing problem solution is less than $10^{-6}$.

\noindent\textbf{Dual values.} To obtain the dual values used in the pricing problem, we do not need to formulate and solve a dual LP as modern LP solvers already keep track of dual values.

%% file: text/experiments.tex
\section{Experimental Evaluation}\label{sec:exp}\input{table_experiments_opp3}

We computationally evaluate the algorithms proposed in \cref{sec:small_support} and \cref{sec:colgen}. We test on the common parametric games shown in \cref{table:rele vs prod}. Appendix~\ref{sec:exp_appendix} provides additional detail about the games.
We ran the experiments on a machine with a 16-core 2.40GHz CPU and 32GB of RAM. We used Gurobi 9.0.3 to solve LPs and MIPs.

{\bf Small-Supported TMECor in Practice. }
\cref{tab:smallsupp} describes the maximum expected utility that the team can obtain by limiting the support of its distribution to $n\in\{1,2,3\}$ semi-randomized correlation plans. Columns denoted by $n=\infty$ show the optimal expected utility of the team at the TMECor (without any limit on the support size). We ran experiments with the opponent as the first (\opppl\,=\,1), second (\opppl\,=\,2), and third player (\opppl\,=\,3) of each game. 
In all the games, distributions with as few as two or three semi-randomized coordination plans gave the team near-optimal expected utility. Moreover, in several games, one or two carefully selected semi-randomized coordination plans are enough to reach an optimal solution. 

\noindent\noindent{\bf Column-Generation in Practice. }
We evaluate our column-generation algorithm against the two prior state-of-the art algorithms for computing a TMECor: the column-generation technique by~\citet{Celli18:Computational} (henceforth CG-18), and the fictitious-team-play algorithm by~\citet{Farina18:Ex} (denoted FTP). Like our algorithm, CG-18 uses column generation approach which lets \opppl{} play sequence-form strategies, while the team's strategy is directly represented as a distribution over joint normal-form plans. On the other hand, FTP is based on the bilinear saddle-point formulation of the problem and is essentially a variation of \emph{fictitious play}~\cite{Brown51:Iterative}. FTP operates on the bilinear formulation of TMECor \eqref{eq:tmecor_basic}: the team and the opponent are treated as two entities that converge to equilibrium in self-play. FTP only guarantees convergence in the limit to an approximate TMECor, while our algorithm certifies optimality. So, the run-time comparison between our algorithm to FTP must be done with care, as the latter never stops, whereas our algorithm and CG-18 terminate after a finite number of iterations with an \emph{exact} optimal strategy. We report the run time of FTP reaching solution quality that is $\epsilon = 50\%$, $10\%$, and $1\%$ off the optimal value (determined by the other two algorithms).
%\footnote{Given a strategy of the team, its exploitability is the difference between the expected utility of a best-responding opponent to the current team strategy, and the expected utility that the  can  with respect to its current strategy.}
%
We set a time limit of 6 hours and a cap of at most four threads for each algorithm.
Table~\ref{tab:colgen} shows the results with the opponent playing as the third player. By \cref{tab:smallsupp}, this is almost always the hardest setting. The results for the other two settings are in~\cref{sec:additional_results}.

%\paragraph{Overall performance} 
Our column-generation algorithm dramatically outperforms FTP and CG-18. There are settings, such as Liar's dice instance [F], where we our algorithm needs just a few seconds to compute an optimal TMECor, while previous algorithms exceed 6 hours. The last column of~\cref{tab:colgen}(c) shows the final team utility. Even when the opponent is playing as the third player, the team is able to reach positive expected utility. 
%The results on the smaller games are in line with those by~\citet{Farina18:Ex}. 
Finally, we identify Liar's dice instance [G] as the current boundary of problem that just cannot be handled with current TMECor technology.

%\paragraph{Benefits from the linear pricer} 
Using the linear relaxation of the pricing problem (``implementation details'' in Section~\ref{sec:implementation_colgen}) often obviated the need to run the slower MIP pricing (see Table~\ref{tab:colgen}(b)).  In all Goofspiel instances (games [D] and [E]) and in small Kuhn poker instances, the MIP pricing is never invoked.

%\paragraph{Benefits from regret-based seeding} 
Regret-based seeding further ameliorates the performance of the algorithm. In the Liar's dice instance [F], it reduced run time by roughly a factor of ten. The value of the initial master solution (that is, before the first pricing) increases significantly 
%WOULD BE NICE TO SAY QUANTITATIVELY ???
with the number of iterations of the no-regret algorithm used for seeding.

%% file: table_experiments_opp3.tex
\begin{table*}
    \sisetup{detect-all=true,scientific-notation=fixed,fixed-exponent=0,round-mode=places,
    round-precision=3}
    \setlength{\tabcolsep}{4.5pt}\small\centering
\begin{tikzpicture}
\node[anchor=south west] at (0, 0) {
    \begin{tabular}{c||rr|rrr|r||rr||rrr|r}
        \toprule
            \multirow{2}{*}{\!\bf Game}  & \multicolumn{2}{c|}{\bf Ours} & \multicolumn{3}{c|}{\bf Fictitious Team Play (FTP)} & \bf \multirow{2}{*}{CG-18} & \multicolumn{2}{c||}{\bf Pricers} & \multicolumn{3}{c|}{\bf Team utility after seeding} & \bf TMECor\\
            &Seeded & \!Not seed.& $\epsilon=50\%$ & $\epsilon=10\%$ & $\epsilon=1\%$ & & Relax.\! & MIP & $m = 1$ &  \num[scientific-notation=false]{1000} & \num[scientific-notation=false]{10000} & \multicolumn{1}{c}{\bf value}\\
        \midrule
            {\bf [A]} & \n{0.001} & \n{0.00100016} & \n{2}$\mathrlap{^\dagger}$ & \n{10}$\mathrlap{^\dagger}$&\n{68}$^\dagger$& \n{0.174723} & 1 & 0 & \num{-0.5000000} & \num{0} & \num{0} &\num{0} \\
            {\bf [B]} & \n{0.001} & \n{0.0339999} & \n{232} & \n{2271} & \tl & \n{26.806077} & 2 & 0 & \num{-0.3645833} & \num{-0.0210305} & \num{-0.0198071} & \num{-0.0417}\\
            {\bf [C]} & \n{17.19899988174}& \n{18.607000} &\n{16957}&\tl&\tl&\tl & 2 & 25 & \num{-0.1545455} & \num{-0.0198071} & \num{-0.0198071} & \num{-0.0140}\\
        \midrule
            {\bf [D]} &\n{0.2669999599}& \n{0.68199992} & \n{50} & \n{561} & \tl & \n{189.39} & 14 & 0 & \num{-0.4359} &  \num{0.252263} & \num{0.252263} & \num{0.2524}\\  
            {\bf [E]} & \n{1.34200000}&\n{1.77499985} & \n{291} &\n{7345} & \tl & \n{1778.109693} & 48 & 0 & \num{-0.8300000} &  \num{0.2481481} & \num{0.250150} & \num{0.2534}\\ 
        \midrule
            {\bf [F]} & \n{101.54700}& \n{682.638000} & \tl &\tl&\tl&\tl & 20 & 7 & \num{-0.4814815} & \num{0.252263} & \num{0.252263} & \num{0.2840} \\  
            {\bf [G]} & \tl & \tl & \tl & \tl & \tl & \tl & \na & \na & \num{-0.6875000} & \num{0.2768617} & \textcolor{gray}{oom} & \na \\ 
        \midrule
            {\bf [H]} & \n{320.518000125885} & \n{353.95799994468}& \tl & \tl & \tl & \tl & 23 & 204 & \num{-2.3544444} & \num{0.0869079} & \num{0.1252301} & \num{0.2926} \\  
            {\bf [I]} & \n{5416.331000089645}& \n{6270.379999876022} & \tl & \tl & \tl & \tl & 5 & 638 & \num{-1.8265306}  & \num{0.0129802} & \num{0.0359574} & \num{0.0850}\\ 
            {\bf [J]}  & \n{668.1010000705719}& \n{889.4049999713898} & \tl & \tl & \tl & \tl & 1232 & 48 & \num{-3.3333333} & \num{0.6455854} & \num{0.6681084} & \num{0.8359}\\ 
        \bottomrule
    \end{tabular}};
    \draw[line width=2.2mm,white] (1.26, 0.08) -- (1.26, 5.34);
    \draw[line width=2.2mm,white] (9.95, 0.08) -- (9.95, 5.34);
    \draw[line width=2.2mm,white] (12.0, 0.08) -- (12.0, 5.34);
    \node[anchor=north] at (5.8,0.1) {\textbf{(a)} --- Comparison of run times};
    \node[anchor=north] at (11.0,0.1) {\textbf{(b)}};
    \node[anchor=north] at (14.5,0.1) {\textbf{(c)}};
\end{tikzpicture}\vspace{-2mm}
\caption{\textbf{(a)} Runtime comparison between our algorithm, FTP, and CG-18. The seeded version of our algorithm runs $m=1000$ iterations of CFR+ (\cref{sec:implementation_colgen}), while the non seeded version runs $m=1$. `{\scriptsize $\dagger$}': since the TMECor value for the game is exactly zero, we measure how long it took the algorithm to find a distribution with expected value at least $-\epsilon/10$ for the team.
\textbf{(b)} Number of times the pricing problem for our column-generation algorithm was solved to optimality by the linear relaxation (`Relax') and by the MIP solver (`MIP') when using our column-generation algorithm.
\textbf{(c)} Quality of the initial strategy of the team obtained for  varying sizes of $S$ compared to the expected utility of the team at the TMECor. `{\textcolor{gray}{oom}}': out of memory. }\vspace{-2mm}
\label{tab:colgen}
\end{table*}

%% file: text/conclusions.tex
\section{Conclusions}
We studied the problem of finding an optimal strategy for a team with two members facing an opponent in an imperfect-information, zero-sum, extensive-form game. We focused on the scenario in which team members are not allowed to communicate during play but can coordinate before the game.  
First, we provided modeling results by drawing a connection to previous results on extensive-form correlation. Then, we developed an algorithm that computes an optimal joint distribution by only using profiles where only one of the team members gets to randomize in each profile. 
We can cap the number of such profiles we allow in the solution. This begets an anytime algorithm by increasing the cap. Moreover, we showed that often a handful of well-chosen such profiles suffices to reach optimal utility for the team. 
Inspired by this observation and leveraging theoretical concepts that we introduced, we developed an efficient column-generation algorithm for finding an optimal strategy for the team. We tested our algorithm on a suite of standard games, showing that it is three order of magnitudes faster than the state of the art and also solves many games that were previously intractable.

% We first provided modeling results about computing a joint optimal distribution by drawing a connection to a different literature on extensive-form correlation.
 
% Second, we developed an algorithm that computes such an optimal distribution by only using profiles where only one of the team members gets to randomize in each profile. 
 
% We can also cap the number of such profiles we allow in the solution. 
 
% This begets an anytime algorithm by increasing the cap. 
 
% We found that often a handful of well-chosen such profiles suffices to reach optimal utility for the team. 
 
% This enables team members to reach coordination through a relatively simple and understandable plan. 
 
% Finally, inspired by this observation and leveraging theoretical concepts that we introduced, we developed an efficient  column-generation algorithm for finding an optimal distribution for the team.
 
% It is three orders of magnitude faster than the prior state of the art on games that the latter can solve and also solves many games that were previously intractable. 

%% file: text/appendix_proofs.tex
\section{Theoretical Details}\label{sec:proofs}

\subsection{Representing Distributions of Play via Extensive-Form Correlation Plans}
As mentioned in the body, every distribution over \emph{randomized} stratregy profiles for the team members is equivalent to a different distribution over \emph{deterministic} strategy profiles by means of Kuhn's theorem~\citep{Kuhn53:Extensive}, one of the most fundamental results about extensive-form game playing. Specifically, given two independent mixed strategies $\vec{y}\tone\in\Y\tone$ and $\vec{y}\ttwo\in\Y\ttwo$ for the team members, let $\mu\tone$ and $\mu\ttwo$ be the distributions over normal-form plans $\Pi\tone, \Pi\ttwo$ equivalent to $\vec{y}\tone$ and $\vec{y}\ttwo$, respectively. Then, the distribution over reandomized strategy profiles that assignes probability 1 to $(\vec{y}\tone,\vec{y}\ttwo)$ is equivalent to the product distribution of $\mu\tone$ and $\mu\ttwo$, that is, the distirbution over $\Pi\tone\times\Pi\ttwo$ that picks a generic profile $(\pi\tone,\pi\ttwo)$ with probability $\pi\tone(\pi\tone)\times\pi\ttwo(\pi\ttwo)$. The reverse is also true: a product distribution over $\Pi\tone\times\Pi\ttwo$ is equivalent to a distribution over randomized profiles that picks exactly one profile with probability $1$.

We now show that a similar result holds when the distribution over normal-form plans is represented as an extensive-form correlation plan. First, we introduce the notion of \emph{product} correlation plan.

\begin{definition}\label{def:product xi}
    Let $\vxi \in \cV$ be a vector in the von Stengel-Forges polytope. We say that $\vxi$ is a \emph{product} correlation plan if
    \[
        \xi\team[\sigma\tone,\sigma\ttwo] = \xi\team[\sigma\tone,\emptyseq]\cdot\xi\team[\emptyseq,\sigma\ttwo]
    \]
    for all $(\sigma\tone,\sigma\ttwo)\in\Sigma\tone \rele\Sigma\ttwo$.
\end{definition}

\begin{lemma}\label{lem:product implies Xi}
    A product correlation plan is always an element of $\Xi\team$.
\end{lemma}
\begin{proof}
    Let $\vxi$ be a product correlation plan. Since by definition, $\vxi\in\cV$, the vectors $\vec{y}\tone,\vec{y}\ttwo$ indexed over $\Sigma\tone$ and $\Sigma\ttwo$, repsectively, and defined as
    \[
        y[\sigma\tone] = \xi\team[\sigma\tone,\emptyseq], y[\sigma\ttwo] = \xi\team[\emptyseq,\sigma\ttwo]
    \]
    are sequence-form strategies. By Kuhn's theorem, there exist distributions $\mu\tone,\mu\ttwo$ over $\Pi\tone$ and $\Pi\ttwo$, respectively, such that
    \begin{align}
        y[\sigma\tone] &= \sum_{\pi\tone \in \Pi\tone(\sigma\tone)} \mu\tone[\pi\tone] \qquad\forall \sigma\tone \in \Sigma\tone,\label{eq:kuhn t1}\\
        y[\sigma\ttwo] &= \sum_{\pi\ttwo \in \Pi\ttwo(\sigma\ttwo)} \mu\ttwo[\pi\ttwo] \qquad\forall \sigma\ttwo \in \Sigma\ttwo\label{eq:kuhn t2}.
    \end{align}
    Consider the distribution $\mu\team$ over $\Pi\tone\times\Pi\ttwo$ defined as the product distribution $\mu\tone\otimes\mu\ttwo$, that is,
    \[
        \mu\team[\sigma\tone,\sigma\ttwo] \defeq \mu\tone[\pi\tone]\cdot\mu\ttwo[\pi\ttwo]
    \]
    for all $(\pi\tone,\pi\ttwo)\in\Pi\tone\times\Pi\ttwo$. We will show that is the extensive-form correlation plan corresponding to $\mu\team$ according to \eqref{eq:xi def}, that is,
    \[
        \xi\team[\sigma\tone,\sigma\ttwo]\defeq \sum_{\substack{\pi\tone\in\Pi\tone(\sigma\tone)\\\pi\ttwo\in\Pi\ttwo(\sigma\ttwo)}} \mu\team[\pi\tone,\pi\ttwo]
    \]
    for all $(\sigma\tone,\sigma\ttwo)\in\Sigma\tone \rele\Sigma\ttwo$. Indeed, using the fact that $\vxi$ is a \emph{product} correlation plan together with~\eqref{eq:kuhn t1} and~\eqref{eq:kuhn t2}:
    \begin{align*}
        \xi\team[\sigma\tone,\sigma\ttwo] &= \xi\team[\sigma\tone,\emptyseq]\cdot\xi\team[\emptyseq,\sigma\ttwo]\\
            &= y\tone[\sigma\tone]\cdot y\ttwo[\sigma\ttwo]\\
            &= \mleft(\sum_{\pi\tone \in \Pi\tone(\sigma\tone)} \mu\tone[\pi\tone]\mright)\mleft(\sum_{\pi\ttwo \in \Pi\ttwo(\sigma\ttwo)} \mu\ttwo[\pi\ttwo]\mright)\\
            &= \sum_{\substack{\pi\tone\in\Pi\tone(\sigma\tone)\\\pi\ttwo\in\Pi\ttwo(\sigma\ttwo)}} \mu\tone[\pi\tone]\cdot \mu\ttwo[\pi\ttwo]\\
            &= \sum_{\substack{\pi\tone\in\Pi\tone(\sigma\tone)\\\pi\ttwo\in\Pi\ttwo(\sigma\ttwo)}} \mu\team[\pi\tone,\pi\ttwo].
    \end{align*}
    This concludes the proof.
\end{proof}

\begin{lemma}\label{lem:product plans to singleton}
    An extensive-form correlation plan is equivalent to a distribution of play for the team that picks one profile of randomized strategies $(\vec{y}\tone,\vec{y}\ttwo)\in\Y\tone\times\Y\ttwo$ if and only if $\vxi$ is a product correlation plan. Furthermore, when that is the case, $y\tone[\sigma\tone] = \xi\team[\sigma\tone,\emptyseq], y\ttwo[\sigma\ttwo] = \xi\team[\emptyseq,\sigma\ttwo]$ for all $\sigma\tone\in\Sigma\tone, \sigma\ttwo\in\Sigma\ttwo$.
\end{lemma}
\begin{proof}
    The proof of \cref{lem:product implies Xi} already shows that when $\vxi$ is a product correlation plan, it is equivalent to playing according to the distribution of play for the team with singleton support $(\vec{y}\tone,\vec{y}\ttwo)$, where $y\tone[\sigma\tone] = \xi\team[\sigma\tone,\emptyseq], y\ttwo[\sigma\ttwo] = \xi\team[\emptyseq,\sigma\ttwo]$ for all $\sigma\tone\in\Sigma\tone, \sigma\ttwo\in\Sigma\ttwo$. So, the only statement that remains to prove is that distributions $\mu\team$ over randomized strategy  profiles for the team members with a singleton support are mapped (\cref{eq:xi def}) to product correlation plans.

    Let $\{(\vec{y}\tone,\vec{y}\ttwo)\} \subseteq \Y\tone\times\Y\ttwo$ be the (singleton) support of $\mu\team$, and let $\mu\tone,\mu\ttwo$ be distributions over $\Pi\tone$ and $\Pi\ttwo$, respectively, equivalent to $\vec{y}\tone$ and $\vec{y}\ttwo$. Then,
    \begin{align}
        y[\sigma\tone] &= \sum_{\pi\tone \in \Pi\tone(\sigma\tone)} \mu\tone[\pi\tone] \qquad\forall \sigma\tone \in \Sigma\tone,\\
        y[\sigma\ttwo] &= \sum_{\pi\ttwo \in \Pi\ttwo(\sigma\ttwo)} \mu\ttwo[\pi\ttwo] \qquad\forall \sigma\ttwo \in \Sigma\ttwo.
    \end{align}
    Since by assumption the two team members sample strategies independently, their equivalent distribution of play over determinitic strategies is the product distribution $\mu\team \defeq \mu\tone \otimes \mu\ttwo$. Using \eqref{eq:xi def}, $\mu\team$ has a representation as extensive-form correlation plan given by
    \begin{align}
        \xi\team[\sigma\tone,\sigma\ttwo] &= \sum_{\substack{\pi\tone\in\Pi\tone(\sigma\tone)\\\pi\ttwo\in\Pi\ttwo(\sigma\ttwo)}} \mu\team[\pi\tone,\pi\ttwo] \nonumber\\
        &= \sum_{\substack{\pi\tone\in\Pi\tone(\sigma\tone)\\\pi\ttwo\in\Pi\ttwo(\sigma\ttwo)}} \mu\tone[\pi\tone]\cdot \mu\ttwo[\pi\ttwo] \nonumber\\
        &= \mleft(\sum_{\pi\tone \in \Pi\tone(\sigma\tone)} \mu\tone[\pi\tone]\mright)\mleft(\sum_{\pi\ttwo \in \Pi\ttwo(\sigma\ttwo)} \mu\ttwo[\pi\ttwo]\mright) \nonumber\\
        &= y\tone[\sigma\tone]\cdot y\ttwo[\sigma\ttwo]\label{eq:prod of ys}
    \end{align}
    for all $(\sigma\tone,\sigma\ttwo)\in\Sigma\tone\times\Sigma\ttwo$. In particular, choosing $\sigma\ttwo = \emptyseq$ in \eqref{eq:prod of ys}, and using the fact that $y\ttwo[\emptyseq] = 1$, we obtain
    \[
        \xi\team[\sigma\tone,\emptyseq] = y\tone[\sigma\tone] \qquad \forall\ \sigma\tone \in \Sigma\tone.
    \]
    Similarly,
    \[
        \xi\team[\emptyseq,\sigma\ttwo] = y\ttwo[\sigma\ttwo] \qquad \forall\ \sigma\ttwo \in \Sigma\ttwo.
    \]
    Substituting the last two equalities into \eqref{eq:prod of ys} we can write
    \[
        \xi\team[\sigma\tone,\sigma\ttwo] = \xi\team[\sigma\tone,\emptyseq]\cdot \xi\team[\emptyseq,\sigma\ttwo]
    \]
    for all $(\sigma\tone,\sigma\ttwo)\in\Sigma\tone\times\Sigma\ttwo$. That, together with the inclusion $\Xi\team \subseteq \cV$, shows that $\vxi$ is a product correlation plan.
\end{proof}

\paragraph{Semi-randomized correlation plans are product plans}

In the body we mentioned that semi-randomized correlation plans correspond to a distribution of play where one team member plays a deterministic strategy and the other team member plays a randomized strategy. We now give more formal grounding that that assertion.

\begin{lemma}\label{lem:semi rand are product}
    Let $\vxi \in \Xi^*\tone \cup \Xi^*\ttwo$ be a semi-randomized plan. Then, $\vxi$ is a product plan.
\end{lemma}

We reuse some ideas that already appeared in \citet{Farina20:Polynomial} to prove \cref{lem:semi rand are product}. In particular, in the proof we will make use of the following lemma.
\begin{lemma}[\citeauthor{Farina20:Polynomial} (\citeyear{Farina20:Polynomial}, Lemma 6)]\label{lem:v zero}
    Let $\vxi \in \cV$. For all $\sigma\tone \in \Sigma\tone$ such that $\xi\team[\sigma\tone, \emptyseq] = 0$, $\xi\team[\sigma\tone, \sigma\ttwo] = 0$ for all $\sigma\ttwo \in \Sigma\ttwo: \sigma\tone\rele\sigma\ttwo$. Similarly, for all $\sigma\ttwo \in \Sigma\ttwo$ such that $\xi\team[\emptyseq, \sigma\ttwo] = 0$, $\xi\team[\sigma\tone, \sigma\ttwo] = 0$ for all $\sigma\tone \in \Sigma\tone: \sigma\tone \rele \sigma\ttwo$.
\end{lemma}

\begin{proof}[Proof of \cref{lem:semi rand are product}]
    We will only show the proof for the case $\vxi \in \Xi^*\tone$. The other case ($\vxi \in \Xi^*\ttwo$) is symmetric.

    To show that
    \[
        \xi\team[\sigma\tone,\sigma\ttwo] =\xi\team[\sigma\tone,\emptyseq]\cdot\xi\team[\emptyseq,\sigma\ttwo]
    \]
    for all $(\sigma\tone,\sigma\ttwo)\in\Sigma\tone\rele\Sigma\ttwo$, we perform induction on the depth of the sequence $\sigma\ttwo$. The depth $\dep(\sigma\ttwo)$ of a generic sequence $\sigma\ttwo = (J,b) \in \Sigma\ttwo$ of Player~$i$ is defined as the number of actions that Player \ttwopl{} plays on the path from the root of the tree down to action $b$ at information set $J$, included. Conventionally, we let the depth of the empty sequence be $0$.

    The base case for the induction proof corresponds to the case where $\sigma\ttwo$ has depth $0$, that is, $\sigma\ttwo = \emptyseq$. In that case, the theorem is clearly true, because $\xi\team[\emptyseq,\emptyseq]= 1$ as part of the von Stengel-Forges constraints (\cref{def:vsf}).

    Now, suppose that the statement holds as long as $\dep(\sigma\ttwo) \le d$. We will show that the statement will hold for any $(\sigma\tone,\sigma\ttwo) \in \Sigma\tone \rele \Sigma\ttwo$ such that $\dep(\sigma\ttwo) \le d+1$. Indeed, consider $(\sigma\tone,\sigma\ttwo) \in \Sigma\tone \rele \Sigma\ttwo$ such that $\sigma\ttwo = (J, b)$ with $\dep(\sigma\ttwo) = d+1$.
    
    There are only two possible cases:
    \begin{itemize}
        \item Case 1: $\xi\team[\emptyseq, \sigma\ttwo] = 0$. From \cref{lem:v zero}, $\xi\team[\sigma\tone, \sigma\ttwo] = 0$ and the statement holds.
        \item Case 2: $\xi\team[\emptyseq, \sigma\ttwo] = 1$.
    From the von Stengel-Forges constraints, $\xi\team[\emptyseq, \sigma(J)] = \sum_{b'\in A_J} \xi\team[\emptyseq, (J, b')] = 1 + \sum_{b' \in A_J, b' \neq b} \xi\team[\emptyseq, (J, b')] \ge 1$. Hence, because all entries of $\xi\team[\emptyseq, \sigma_2]$ are in $\{0,1\}$ by definition of $\Xi^*\tone$, it must be $\xi\team[\emptyseq, \sigma(J)] = 1$ and $\xi\team[\emptyseq, (J,b')] = 0$ for all $b' \in A_J, b' \neq b$.

    Using the inductive hypothesis, we have that
    \begin{equation}\label{eq:amended}
        \xi\team[\sigma\tone, \sigma(J)] = \xi\team[\sigma\tone,\emptyseq]\cdot\xi\team[\emptyseq,\sigma(J)] = \xi\team[\sigma\tone,\emptyseq]
    \end{equation}
    for all $\sigma\tone \in \Sigma\tone, \sigma\tone\rele\sigma(J)$.
    On the other hand, since $\xi\team[\emptyseq, (J,b')] = 0$ for all $b' \in A_J, b' \neq b$, from \cref{lem:v zero} we have that
    \begin{equation}\label{eq:lem5 step2}
        \xi\team[\sigma\tone, (J,b')] = 0 \quad\forall \sigma\tone \rele J, b' \neq b.
    \end{equation}
    Hence, summing over all $b' \in A_J$ and using the von Stengel-Forges constraints, we get
    \begin{align*}
         \xi\team[\sigma\tone,\emptyseq]\cdot \xi\team[\emptyseq, \sigma\ttwo]
        &= \xi\team[\sigma\tone, \sigma(J)] \\
        &= \sum_{b' \in A_J} \xi\team[\sigma\tone,(J, b')]\\
        &= \xi\team[\sigma\tone, (J, b)] = \xi\team[\sigma\tone, \sigma\ttwo]
    \end{align*}
    for all $\sigma\tone \rele (J, b)$. This concludes the proof by induction.\qedhere
    \end{itemize}
\end{proof}

So, from \cref{lem:product plans to singleton} it follows that semi-randomized plans correspond to distributions of play over randomized profiles with the singleton support $(\vec{y}\tone,\vec{y}\ttwo) \in \Y\tone\times\Y\ttwo$. Furthermore, because of the second part of \cref{lem:product plans to singleton}, when $\vxi\in\Xi\tone^*$, $\vec{y}\ttwo[\sigma\ttwo] \in \{0,1\}$ for all $\sigma\ttwo \in \Sigma\ttwo$, which means that $\vec{y}\ttwo$ is a deterministic strategy for Player \ttwopl{} (a similar statement holds for $\vxi\in\Xi^*\ttwo$).

\paragraph{Convex combinations of product plans}
Both of the algorithms we presented in the paper ultimately produce an extensive-form correlation plan $\vxi$ that is a convex combination of semi-randomized plans $\vxi^{(1)}, \dots, \vxi^{(n)}$, that is, of the form
\[
    \vxi  = \lambda^{(1)} \vxi^{(1)} + \cdot + \lambda^{(n)}\vxi^{(n)}
\]
for $\lambda^{(i)} \ge 0$ such that $\lambda^{(1)} + \dots + \lambda^{(n)} = 1$. Since semi-randomized correlation plans are product correlation plans (\cref{lem:semi rand are product}), from \cref{lem:product plans to singleton} each $\vxi^{(i)}$ is equivalent to the team playing a single profile of randomized strategies $(\vec{y}\tone^{(i)}, \vec{y}\ttwo^{(i)}) \in \Y\tone\times\Y\ttwo$ with probability $1$. By linearity, it is immediate to show that $\vxi$ is equivalent to playing according to the distribution over randomized strategies for the team that picks $(\vec{y}\tone^{(i)}, \vec{y}\ttwo^{(i)})$ with probability $\lambda^{(i)}$.

\subsection{TMECor Formulation Based on Extensive-Form Correlation Plans}
\proptmecorlp*
\begin{proof}
    We follow the steps mentioned in the body, starting from the bilinear saddle point problem formulation of the problem of computing a TMECor strategy for the team:
    \[
        \argmax_{\vec{\xi}\team \in \Xi\team} \min_{\vec{y}\opp \in \Y\opp} \sum_{z\in Z} \hat{u}\team(z) \xi\team[\sigma\tone(z),\sigma\ttwo(z)] y[\sigma\opp(z)].
    \]
    Expanding the constraint $\vec{y}\opp \in \Y\opp$ using the \emph{sequence-form constraints}~\citep{Koller96:Efficient,Stengel96:Efficient}, the inner minimization problem is
    \begin{equation*}
        (P): \mleft\{\hspace{-1.25mm}\begin{array}{l}
            \displaystyle
            %% OBJECTIVE
            \min_{\vec{y}\opp}~~~ \sum_{z\in Z} \hat{u}\team(z) \xi\team[\sigma\tone(z),\sigma\ttwo(z)] y[\sigma\opp(z)]\\
            %\text{\normalfont subject to:}\\
            %% CONSTRAINT 1
            ~\circled{\normalfont 1}~ \displaystyle -y[\sigma(I)] + \sum_{a \in A_I} y\opp[(I,a)] = 0 \qquad\forall I \in \I\opp\\
            ~\circled{\normalfont 2}~ \displaystyle y\opp[\emptyseq] = 1\\[2mm]
            ~\circled{\normalfont 3}~ \displaystyle y\opp[\sigma\opp] \ge 0 \qquad\forall\ \sigma\opp\in\Sigma\opp.
        \end{array}\mright.
    \end{equation*}
    Introducing the free dual variables $\{v_I : I\in \I\opp\}$ for Constraint \circled{1}, and the free dual variable $v_\emptyseq$ for Constraint \circled{2}, we obtain the dual linear program
    \begin{equation*}
        (D): \mleft\{\hspace{-1.25mm}\begin{array}{l}
            \displaystyle
            %% OBJECTIVE
            \max_{v_I, v_\emptyseq}~~~ v_\emptyseq,
            \quad\text{\normalfont subject to:}\\
            ~\circled{\normalfont 1}~ \displaystyle v_I - ~\sum_{\mathclap{\substack{I' \in \mathcal{I}\opp\\\sigma\opp(I') = (I, a)}}}~ v_{I'} \le ~\sum_{\mathclap{\substack{z \in Z\\\sigma\opp(z) = (I,a)}}} \hat{u}\team(z)\xi\team[\sigma\tone(z),\sigma\ttwo(z)]\\[-4mm]
            \hspace{5cm}\forall\, (I,\!a) \!\in\! \Sigma\opp\!\!\setminus\! \{\emptyseq\}\\[2mm]
            %% CONSTRAINT 2
            ~\circled{\normalfont 2}~\displaystyle v_\emptyseq - ~\sum_{\mathclap{\substack{I' \in \mathcal{I}\opp\\\sigma\opp(I') = \emptyseq}}}~ v_{I'} \le ~\sum_{\mathclap{\substack{z \in Z\\\sigma\opp(z) = \emptyseq}}} \hat{u}\team(z)\xi\team[\sigma\tone(z),\sigma\ttwo(z)]\\[8mm]
            %% CONSTRAINT 3
            ~\circled{\normalfont 3}~v_\emptyseq\textnormal{ free}, v_{I} \textnormal{ free } \quad \forall\ I \in \mathcal{I}\opp.
        \end{array}\mright.
    \end{equation*}
    So, $\vxi$ is a TMECor if and only if it is a solution of $\argmax_{\vxi\in\Xi\team} (D)$, which is exactly the statement.
\end{proof}

\subsection{Semi-Randomized Correlation Plans}

\propconvexhull*
\begin{proof}
    We will show that $\Xi\team = \co \Xi\tone^*$. The proof that $\Xi\team = \co\Xi\ttwo^*$ is symmetric.

    We will break the proof of $\Xi\team = \co\Xi\tone^*$ into two parts:
    \begin{itemize}[leftmargin=10mm]
        \item[($\subseteq$)] In the first part of the proof, we argue that $\Xi\tone^* \subseteq \Xi\team$. This is straightforward: from \cref{lem:semi rand are product} we know that all elements of $\Xi\tone^*$ are product correlation plans (\cref{def:product xi}), which implies that $\Xi\tone^*\subseteq \Xi\team$ by \cref{lem:product implies Xi}. Since convex hulls preserve inclusions, we have
        \[
            \co\Xi\tone^* \subseteq \co\Xi\team,
        \]
        which is exactly the statement $\Xi\tone^* \subseteq \Xi\team$ upon using the known fact that $\Xi\team$ is a convex polytope and therefore $\co\Xi\team = \Xi\team$.

        \item[($\supseteq$)] To complete the proof, we now argue that the reverse inclusion, namely $\Xi\team \subseteq \co\Xi\tone^*$, also holds. Let $f : \mu\team \mapsto \xi\team$ be the mapping from the distribution of play $\mu\team\in\Delta(\Pi\tone\times\Pi\ttwo)$ to its corresponding extensive-form correlation plan defined in \cref{eq:xi def}. By definition, $\Xi\team = f(\Delta(\Pi\tone\times\Pi\ttwo))$. Let $\one_{(\pi\tone,\pi\ttwo)}$ denote the distribution of play with singleton support $(\pi\tone,\pi\ttwo)$, that is, the distribution of play that assigns the deterministic strategy profile $(\pi\tone,\pi\ttwo)$ for the team with probability $1$. Since $f$ is linear, and since
        \[
            \Delta(\Pi\tone\times\Pi\ttwo) = \co\{\one_{(\pi\tone,\pi\ttwo)}:\pi\tone\in\Pi\tone, \pi\ttwo\in\Pi\ttwo\},
        \]
        we have 
        \[
            \Xi\team = \co\{f(\one_{(\pi\tone,\pi\ttwo)}):\pi\tone\in\Pi\tone, \pi\ttwo\in\Pi\ttwo\}.
        \]
        Hence, to conclude the proof of this part, it will be enough to show that for each $\pi\tone\in\Pi\tone, \pi\ttwo\in\Pi\ttwo$, it holds that $f(\one_{(\pi\tone,\pi\ttwo)}) \in \Xi\tone^*$. Since $\one_{(\pi\tone,\pi\ttwo)}$ assigns probability $1$ to one profile and $0$ to all other profiles, $f(\one_{(\pi\tone,\pi\ttwo)})$ is an extensive-form correlation plan whose entris are all in $\{0,1\}$. So, in particular, $f(\one_{(\pi\tone,\pi\ttwo)}) \in \Xi^*\tone$. This concludes the proof of the inclusion $\Xi\team \subseteq \co\Xi\tone^*$.
    \end{itemize}
    Together, the two statements that we just prove show that $\Xi\team = \co\Xi\tone^*$.

    Finally, using the fact that unions and convex hulls commute, we have
    \[
        \co (\Xi\tone^* \cup \Xi\ttwo^*) = (\co\Xi\tone^*)\cup(\co\Xi\ttwo^*) = \Xi\team \cup \Xi\team = \Xi\team,
    \]
    thereby concluding the proof.
\end{proof}

%% file: text/appendix_games.tex
\section{Game Instances}\label{sec:exp_appendix}

The size of the parametric instances we use as benchmark is
described in \cref{table:rele vs prod}. In the following, we provide a detailed explanation of the rules of each game.

\paragraph{Kuhn poker} Two-player Kuhn poker was originally proposed by \citet{Kuhn50:Simplified}. We employ the three-player variation described in \citet{Farina18:Ex}. In a three-player Kuhn poker game with rank $r$ there are $r$ possible cards. At the beginning of the game, each player pays one chip to the pot, and each player is dealt a single private card. The first player can check or bet, i.e., putting an additional chip in the pot. Then, the second player can check or bet after a first player's check, or fold/call the first player's bet. If no bet was previously made, the third player can either check or bet. Otherwise, the player has to fold or call. After a bet of the second player (resp., third player), the first player (resp., the first and the second players) still has to decide whether to fold or to call the bet. At the showdown, the player with the highest card who has not folded wins all the chips in the pot.

\paragraph{Goofspiel} This bidding game was originally introduced by \citet{Ross71:Goofspiel}. We use a 3-rank variant, that is, each player has a hand of cards with values $\{-1, 0, 1\}$. A third stack of cards with values $\{-1,0,1\}$ is shuffled and placed on the table. At each turn, a prize card is revealed, and each player privately chooses one of his/her cards to bid, with the highest card winning the current prize. In case of a tie, the prize is split evenly among the winners. After 3 turns, all the prizes have been dealt out and the payoff of each player is computed as follows: each prize card’s value is equal to its face value and the players’ scores are computed as the sum of the values of the prize cards they have won.

\paragraph{Goofspiel with limited information} This is a variant of Goofspiel introduced by \citet{Lanctot09:Monte}. In this variation, in each turn the players do not reveal the cards that they have played. Rather, they show their cards to a fair umpire, which determines which player has played the highest card and should therefore received the prize card. In case of tie, the umpire directs the players to split the prize evenly among the winners, just like in the Goofspiel game. This makes the game strategically more challenging as players have less information regarding previous opponents’ actions. 

\paragraph{Leduc poker} We use a three-player version of the classical Leduc hold'em poker introduced by \citet{Southey05:Bayes}. We employ game instances of rank 3, in which the deck consists of three suits with 3 cards each. Our instances are parametric in the maximum number of bets, which in limit hold'em is not necessarely tied to the number of players. The maximum number of raise per betting round can be either 1, 2 or 3. As the game starts players pay one chip to the pot. There are two betting rounds. In the first one a single private card is dealt to each player while in the second round a single board card is revealed. The raise amount is set to 2 and 4 in the first and second round, respectively.

\paragraph{Liar's dice} Liar's dice is another standard benchmark introduced by \citet{Lisy15:Online}. In our three-player implementation, at the beginning of the game each of the three players privately rolls an unbiased $k$-face die. Then, the three players alternate in making (potentially false) claims about their toss. The first player begins bidding, announcing any face value up to $k$ and the minimum number of dice that the player believes are showing that value among the dice of all the players. Then, each player has two choices during their turn: to make a higher bid, or to challenge the previous bid by declaring the previous bidder a "liar". A bid is higher than the previous one if either the face value is higher, or the number of dice is higher. If the current player challenges the previous bid, all dice are revealed. If the bid is valid, the last bidder wins and obtains a reward of +1 while the challenger obtains a negative payoff of -1. Otherwise, the challenger wins and gets reward +1, and the last bidder obtains reward of -1. All the other players obtain reward 0. We test our algorithms on Liar's dice instances with $k=3$ and $k=4$.

%% file: text/appendix_experiments.tex
\section{Additional Experimental Results}\label{sec:additional_results}

All experiments were run 10 times, and the experimental tables show average run times. We always use the same random seed to sample no-regret strategies for the team members in the seeding phase of our column-generation algorithm. The seed was never changed, and we don't treat it as a hyperparameter. So, all algorithms are deterministic, and the only source of randomness in the run time is due to system load. Consequently, we observed small standard deviations in the run times, less than $10\%$ in all cases.

We used the same time limit for FTP that was found to be beneficial by the original authors~\citep{Farina18:Ex}, namely $15$ seconds.
For FTP and CG-18, we used the original implementations, with permission from the authors. In all algorithms, we observed that the majority of time is spent within Gurobi.

Table \ref{tab:colgen opp 1} and Table \ref{tab:colgen opp 2} show the comparison between our column-generation algorithm, FTP, and CG-18 when the opponent plays as the first and as the second player, respectively.

\input{table_experiments_opp1}
\input{table_experiments_opp2}

\subsection*{Comparison between the Algorithm of \cref{sec:small_support} and the Prior State of the Art}

Depending on the cap $n$ on the number or semi-randomized correlation plans, the algorithm we describe in \cref{sec:small_support} might not reach the optimal TMECor value for the team (although, as we argue in \cref{sec:exp}, a very small $n$ already guarantees a large fraction of the optimal value empirically). 

For completeness, we report the run time of the algorithm for a sample instance. We employ instance [H] with $\opppl=3$ as it is has a good trade-off between dimensions and manageability. 
When $n=1$ the algorithm reaches an optimal solution in 9.74s. The optimal solution with $n=1$ achieves $63\%$ of the optimal utility with no restrictions on the number of plans. With $n=2$ the run time is 5m38s and the solution reaches $84\%$ of the optimal value. 

The column-generation algorithm has better run time performances and guarantees to reach an optimal solution without having to pick the right support size. However, we observe that the algorithm of \cref{sec:small_support} already outperforms FTP and CG-18. Specifically, FTP cannot reach a strategy guaranteeing $50\%$ of the optimal utility within the time limit, while our algorithm guarantees $84\%$ of the optimal value within roughly 5 minutes. On the other hand, CG-18 cannot complete even a single iteration within the time limit. This confirms the our pricing formulation is significantly tighter than previous formulations.

%% file: table_experiments_opp1.tex
\begin{table*}
    \sisetup{detect-all=true,scientific-notation=fixed,fixed-exponent=0,round-mode=places,
    round-precision=3}
    \setlength{\tabcolsep}{4.5pt}\small\centering
\begin{tikzpicture}
\node[anchor=south west] at (0, 0) {
    \begin{tabular}{c||rr|rrr|r||rr||rrr|r}
        \toprule
            \multirow{2}{*}{\!\bf Game}  & \multicolumn{2}{c|}{\bf Ours} & \multicolumn{3}{c|}{\bf Fictitious Team Play (FTP)} & \bf \multirow{2}{*}{CG-18} & \multicolumn{2}{c||}{\bf Pricers} & \multicolumn{3}{c|}{\bf Team utility after seeding} & \bf TMECor\\
            &Seeded & \!Not seed.& $\epsilon=50\%$ & $\epsilon=10\%$ & $\epsilon=1\%$ & & Relax.\! & MIP & $m = 1$ &  \num[scientific-notation=false]{1000} & \num[scientific-notation=false]{10000} & \multicolumn{1}{c}{\bf value}\\
        \midrule
            {\bf [A]} & \n{0.0016667048} & \n{0.000666658} & \n{0.0}$\mathrlap{^\dagger}$ & \n{15.0}$\mathrlap{^\dagger}$ & \n{155.0}$^\dagger$& \n{0.066486} & 5 & 0 & \num{-0.5666667} & \num{-0.1333333} & \num{-0.1333333} & \num{0} \\
            {\bf [B]} & \n{0.0208333333} & \n{0.00316663} & \n{0} & \n{999.0} & \tl & \n{1.006949} & 0 & 3 & \num{-0.375} & \num{0.0368957} & \num{0.0378788} & \num{0.0379}\\
            {\bf [C]} & \n{5.692666689554851} & \n{5.79283328851064} & \n{456.0} & \tl & \tl & \tl & 8 & 41 & \num{-0.1660857} & \num{0.05834519999999999} & \num{0.058494800000000007} & \num{0.0664}\\
        \midrule
            {\bf [D]} & \n{0.18599998950958252} & \n{0.3039999802907308} & \n{0.0023} & \tl & \tl & \n{116.941626} & 19 & 0 & \num{-0.492424} & \num{0.2514161} & \num{0.2516204} & \num{0.2524}\\
            {\bf [E]} & \n{0.46383337179819745} & \n{0.859500010808309} & \n{0.0064} & \tl & \tl & \n{1397.893382} & 33 & 0 & \num{-1.0} & \num{0.2494477} & \num{0.2534392} & \num{0.2534}\\
        \midrule
            {\bf [F]} & \n{2.139} & \n{4.209} & \n{1165.0} & \tl & \tl & \tl & 1 & 0 & \num{-0.7481} & \num{0} & \num{0} & \num{0}\\
            {\bf [G]} & \n{71.340} & \n{2483.55} & \tl & \tl & \tl & \tl & 0 & 2 & \num{-0.721154} & \num{0.0625} & \textcolor{gray}{oom} & \num{0.0625}\\ 
        \midrule
            {\bf [H]} & \n{43.71649992465973} & \n{84.6965000629425} & \n{10144.0} & \tl & \tl & \tl & 9 & 79 & \num{-3.1419192} & \num{0.2103755} & \num{0.2284915} & \num{0.2765}\\
            {\bf [I]} & \n{2638.9030001163483} & \n{2768.3324999809265} & \tl & \tl & \tl & \tl & 0 & 614 & \num{-3.0906634} & \num{0.1111782} & \num{0.1217689} & \num{0.1422}\\
            {\bf [J]} & \n{228.12250006198883} & \n{696.7568332751592} & \tl & \tl & \tl & \tl & 1612 & 37 & \num{-4.0} & \num{0.6268873} & \num{0.650576} & \num{0.8359}\\
        \bottomrule
    \end{tabular}};
    \draw[line width=2.2mm,white] (1.26, 0.08) -- (1.26, 5.34);
    \draw[line width=2.2mm,white] (9.95, 0.08) -- (9.95, 5.34);
    \draw[line width=2.2mm,white] (12.0, 0.08) -- (12.0, 5.34);
    \node[anchor=north] at (5.8,0.1) {\textbf{(a)} --- Comparison of run times};
    \node[anchor=north] at (11.0,0.1) {\textbf{(b)}};
    \node[anchor=north] at (14.5,0.1) {\textbf{(c)}};
\end{tikzpicture}\vspace{-2mm}
\caption{
\textbf{Results for $\opppl=1$.}
\textbf{(a)} Runtime comparison between our algorithm, FTP, and CG-18. The seeded version of our algorithm runs $m=1000$ iterations of CFR+ (\cref{sec:implementation_colgen}), while the non seeded version runs $m=1$. `{\scriptsize $\dagger$}': since the TMECor value for the game is exactly zero, we measure how long it took the algorithm to find a distribution with expected value at least $-\epsilon/10$ for the team.
\textbf{(b)} Number of times the pricing problem for our column-generation algorithm was solved to optimality by the linear relaxation (`Relax') and by the MIP solver (`MIP') when using our column-generation algorithm (seeded version with $m=1000$).
\textbf{(c)} Quality of the initial strategy of the team obtained for  varying sizes of $S$ compared to the expected utility of the team at the TMECor. `{\textcolor{gray}{oom}}': out of memory. }\vspace{-2mm}
\label{tab:colgen opp 1}
\end{table*}

%% file: table_experiments_opp2.tex
\begin{table*}
    \sisetup{detect-all=true,scientific-notation=fixed,fixed-exponent=0,round-mode=places,
    round-precision=3}
    \setlength{\tabcolsep}{4.5pt}\small\centering
\begin{tikzpicture}
\node[anchor=south west] at (0, 0) {
    \begin{tabular}{c||rr|rrr|r||rr||rrr|r}
        \toprule
            \multirow{2}{*}{\!\bf Game}  & \multicolumn{2}{c|}{\bf Ours} & \multicolumn{3}{c|}{\bf Fictitious Team Play (FTP)} & \bf \multirow{2}{*}{CG-18} & \multicolumn{2}{c||}{\bf Pricers} & \multicolumn{3}{c|}{\bf Team utility after seeding} & \bf TMECor\\
            &Seeded & \!Not seed.& $\epsilon=50\%$ & $\epsilon=10\%$ & $\epsilon=1\%$ & & Relax.\! & MIP & $m = 1$ &  \num[scientific-notation=false]{1000} & \num[scientific-notation=false]{10000} & \multicolumn{1}{c}{\bf value}\\
        \midrule
            {\bf [A]} & \n{0.0} & \n{0.001000086466471354} & \n{0.0}$\mathrlap{^\dagger}$ & \n{19}$\mathrlap{^\dagger}$& \n{189}$^\dagger$& \n{0.147337} & 1 & 0 & \num{-0.6333333} & \num{0.0} & \num{0.0} & \num{0}\\
            {\bf [B]} & \n{0.0} & \n{0.011333346366882324} & \n{99} & \tl & \tl & \n{7.527870} & 1 & 0 & \num{-0.25} & \num{ 0.0265152} & \num{0.0265152} & \num{0.0265} \\
            {\bf [C]} & \n{6.474833329518636} & \n{5.635833303133647} & \n{2888.0} & \tl & \tl & \tl & 6 & 33 & \num{-0.1260534} & \num{0.0267577} & \num{0.033310300000} & \num{0.0380} \\
        \midrule
            {\bf [D]} & \n{0.1443333625793457} & \n{0.36766664187113446} & \n{0.00087} & \tl & \tl & \n{106.067944} & 14 & 0 & \num{-0.384469} & \num{0.2516204} & \num{0.2524217} & \num{0.2524} \\
            {\bf [E]} & \n{0.640999992688497} & \n{0.9039999643961588} & \n{1.38834} & \tl & \tl & \n{750.767358} & 40 & 0 & \num{-3.0} & \num{0.252222199999999} & \num{0.252282} & \num{0.2534} \\
        \midrule
            {\bf [F]} & \n{55.00099992752075} & \n{539.7705000638962} &\n{5455.0} & \tl & \tl & \tl & 21 & 0 & \num{-0.6296296} & \num{0.2561728} & \num{ 0.2561728} & \num{0.2561728} \\
            {\bf [G]} & \tl & \tl & \tl & \tl & \tl & \tl & \na & \na & \num{-0.765625} & \num{0.2640625} & \textcolor{gray}{oom} & \na \\
        \midrule
            {\bf [H]} & \n{450.69799995422363} & \n{485.9849998950958} & \tl & \tl & \tl & \tl & 25 & 335 & \num{-2.0017094} & \num{0.1767536} & \num{0.2008117} & \num{0.3450} \\
            {\bf [I]} & \n{3452.269499897957} & \n{4148.50100004673} & \tl & \tl & \tl & \tl & 1 & 492 & \num{-2.5046703} & \num{0.0958098} & \num{0.109763} & \num{0.1420} \\
            {\bf [J]} & \n{431.86300003528595} & \n{316.2300001382828} & \tl & \tl & \tl & \tl & 2508 & 37 & \num{-7.5} & \num{0.6295301} & \num{ 0.818981} & \num{0.9709} \\
        \bottomrule
    \end{tabular}};
    \draw[line width=2.2mm,white] (1.26, 0.08) -- (1.26, 5.34);
    \draw[line width=2.2mm,white] (9.95, 0.08) -- (9.95, 5.34);
    \draw[line width=2.2mm,white] (12.0, 0.08) -- (12.0, 5.34);
    \node[anchor=north] at (5.8,0.1) {\textbf{(a)} --- Comparison of run times};
    \node[anchor=north] at (11.0,0.1) {\textbf{(b)}};
    \node[anchor=north] at (14.5,0.1) {\textbf{(c)}};
\end{tikzpicture}\vspace{-2mm}
\caption{
\textbf{Results for $\opppl=2$.}
\textbf{(a)} Runtime comparison between our algorithm, FTP, and CG-18. The seeded version of our algorithm runs $m=1000$ iterations of CFR+ (\cref{sec:implementation_colgen}), while the non seeded version runs $m=1$. `{\scriptsize $\dagger$}': since the TMECor value for the game is exactly zero, we measure how long it took the algorithm to find a distribution with expected value at least $-\epsilon/10$ for the team.
\textbf{(b)} Number of times the pricing problem for our column-generation algorithm was solved to optimality by the linear relaxation (`Relax') and by the MIP solver (`MIP') when using our column-generation algorithm (seeded version with $m=1000$).
\textbf{(c)} Quality of the initial strategy of the team obtained for  varying sizes of $S$ compared to the expected utility of the team at the TMECor. `{\textcolor{gray}{oom}}': out of memory. }\vspace{-2mm}
\label{tab:colgen opp 2}
\end{table*}